\documentclass[letterpaper,11pt]{article}

\usepackage{fullpage}
\usepackage{amssymb,amsmath,amsthm,graphicx,algorithm,algorithmic}

\newtheorem{thm}{Theorem}
\newtheorem{lem}[thm]{Lemma}

\newtheorem{prop}[thm]{Proposition}

\theoremstyle{definition}
\newtheorem{defn}[thm]{Definition}

\newcommand{\abs}[1]{\lvert #1 \rvert}
\newcommand{\norm}[1]{\lVert #1 \rVert}

\newcommand{\st}{\;:\;}                         
\newcommand{\ve}[2]{\langle #1 ,  #2 \rangle}   
\newcommand{\eqdef}{\stackrel{\text{{\tiny def}}}{=}}
\newcommand{\J}{\Omega}

\newcommand{\calP}{{\mathcal{P}}}

\newcommand{\hphi}{g}

\newcommand{\E}{\mathbb{E}}
\newcommand{\R}{\mathbb{R}}
\newcommand{\Prob}{\mathbf{P}}
\newcommand{\Exp}{\mathbf{E}}
\newcommand{\Q}{Q}
\newcommand{\z}{z}
\newcommand{\uu}{u}

\newcommand{\ff}{f}
\newcommand{\FF}{F}
\newcommand{\OO}{\Psi}

\DeclareMathOperator{\ESO}{ESO}
\DeclareMathOperator{\dom}{dom}
\DeclareMathOperator{\nnz}{nnz}

\begin{document}

\title{Smooth Minimization of Nonsmooth Functions \\with Parallel Coordinate Descent Methods}

\author{Olivier Fercoq \footnote{School of Mathematics, The University of Edinburgh, United Kingdom (e-mail: olivier.fercoq@ed.ac.uk)} \qquad \qquad  Peter Richt\'{a}rik \footnote{School of Mathematics, The University of Edinburgh, United Kingdom (e-mail: peter.richtarik@ed.ac.uk) \qquad
The work of both authors was supported by the EPSRC grant EP/I017127/1 (Mathematics for Vast Digital Resources). The work of P.R.\ was also supported by the Centre for Numerical Algorithms and Intelligent Software (funded by EPSRC grant EP/G036136/1 and the Scottish Funding Council).}}

\date{September 22, 2013}

\maketitle

\begin{abstract}
We study the performance of a family of randomized parallel coordinate descent methods for minimizing the sum of a nonsmooth and separable convex functions.
The problem class  includes as a special case L1-regularized L1 regression and the minimization of the exponential loss (``AdaBoost problem''). We assume the input data defining the loss function is contained in a sparse $m\times n$ matrix $A$ with at most $\omega$ nonzeros in each row. Our methods need $O(n \beta/\tau)$ iterations to find an approximate solution with high probability, where $\tau$ is the number of processors and $\beta  = 1 + (\omega-1)(\tau-1)/(n-1)$ for the fastest variant. The notation hides dependence on quantities such as the required accuracy and confidence levels and the distance of the starting iterate from an optimal point. Since $\beta/\tau$ is a decreasing function of $\tau$, the method needs fewer iterations when more processors are used. Certain variants of our algorithms perform on average only $O(\nnz(A)/n)$ arithmetic operations during a single iteration per processor and, because $\beta$ decreases when $\omega$ does, fewer iterations are needed for sparser problems.
\end{abstract}

\section{Introduction}

It is increasingly common that practitioners in machine learning, optimization, biology, engineering and various industries need to solve optimization problems with  number of variables/coordinates so huge that classical algorithms, which for historical reasons almost invariably focus on obtaining solutions of high accuracy, are not efficient enough, or are outright unable to perform even a single iteration. Indeed, in the \emph{big data optimization} setting, where the number $N$ of variables is huge, inversion of matrices is not possible, and even operations such as matrix vector multiplications are too expensive. Instead, attention is shifting towards simple methods, with cheap iterations, low memory requirements and good parallelization and scalability properties.


If the accuracy requirements are moderate and the problem  has only simple constraints (such as box constraints), methods with these properties do exist: \emph{parallel coordinate descent methods}  \cite{Bradley:PCD-paper, RT:TTD2011, RT:PCDM, minibatch-ICML2013} emerged as a very promising class of algorithms in this domain.

\subsection{Parallel coordinate descent methods}

In a recent paper \cite{RT:PCDM}, Richt\'{a}rik and Tak\'{a}\v{c} proposed and studied the complexity of a \emph{parallel coordinate descent method (PCDM)} applied to the convex composite\footnote{Gradient methods for problems of this form were studied by Nesterov~\cite{Nesterov:2007composite}.} optimization problem \begin{equation}\label{eq:P0}\min_{x\in \R^N} \phi(x) + \Psi(x),\end{equation}
where $\phi: \R^N \to \R$ is an \emph{arbitrary differentiable} convex function and $\Psi: \R^N\to \R\cup \{+\infty\}$ is a simple (block) \emph{separable} convex regularizer, such as $\lambda \|x\|_1$. The $N$ variables/coordinates of $x$ are assumed to be partitioned into $n$ blocks, $x^{(1)}, x^{(2)}, \dots, x^{(n)}$ and PCDM at  each iteration computes  and applies updates to a randomly chosen subset $\hat{S} \subseteq [n]\eqdef\{1,2,\dots,n\}$ of blocks (a ``sampling'') of the decision vector, \emph{in parallel}. Formally, $\hat{S}$ is a random set-valued mapping with values in $2^{[n]}$.

PCDM encodes a family of algorithms where each variant is characterized by the probability law governing $\hat{S}$. The sets generated throughout the iterations are assumed to be independent and identically distributed.  In this paper we focus on \emph{uniform samplings}, which are characterized by the requirement that
$\Prob(i \in \hat{S}) = \Prob(j \in \hat{S})$ for all $i,j \in [n]$. It is easy to see that for a uniform sampling one necessarily has\footnote{This and other identities for block samplings were derived in \cite[Section 3]{RT:PCDM}.} \begin{equation}\label{eq:uniform_samp_basic}\Prob(i \in \hat{S}) = \frac{\Exp[|\hat{S}|]}{n}.\end{equation}
In particular, we will focus on two special classes of uniform samplings: i) those for which $\Prob(|\hat{S}|=\tau)=1$ (\emph{$\tau$-uniform samplings}), and ii) $\tau$-uniform saplings with the additional property that all subsets of cardinality $\tau$ are chosen equally likely (\emph{$\tau$-nice samplings}). We will also say that a sampling is \emph{proper} if $\Prob(|\hat{S}|\geq 1) > 0$.

It is clearly important to understand whether choosing $\tau>1$, as opposed to $\tau=1$, leads to acceleration in terms of an improved complexity bound.
Richt\'{a}rik and Tak\'{a}\v{c}~\cite[Section~6]{RT:PCDM} established \emph{generic iteration complexity results} for PCDM applied to \eqref{eq:P0}---we describe them in some detail in Section~\ref{sec:generic}. Let us only mention now that these results are generic in the sense that they hold under the blanket assumption that a certain inequality  involving $\phi$ and $\hat{S}$ holds, so that if one is able to derive this inequality for a certain class of smooth convex functions $\phi$, complexity results are readily available. The inequality (called Expected Separable Overapproximation, or ESO) is
\begin{equation}\label{eq:ESO}\Exp\left[\phi(x + h_{[\hat{S}]})\right] \leq \phi(x) + \frac{\Exp[|\hat{S}|]}{n}\left(\ve{\nabla \phi(x)}{h} + \frac{\beta}{2}\sum_{i=1}^n w_i \ve{B_i h^{(i)}}{h^{(i)}}\right), \qquad x,h \in \R^N,\end{equation}
where $B_i$ are positive definite matrices (these can be chosen based on the structure of $\phi$, or simply taken to be identities), $\beta>0$,  $w=(w_1,\dots,w_n)$ is a vector of positive weights, and  $h_{[\hat{S}]}$ denotes the random vector in $\R^N$ obtained from $h$ by zeroing out all its blocks that do not belong to $\hat{S}$. That is,  $h_{[S]}$ is the vector in $\R^N$ for which $h_{[S]}^{(i)}=h^{(i)}$ if $i \in S$ and $h_{[S]}^{(i)}=0$, otherwise. When \eqref{eq:ESO} holds, we say that $\phi $ admits a $(\beta,w)$-ESO with respect
to $\hat{S}$. For simplicity, we may sometimes write $(\phi, \hat{S})\sim \ESO(\beta,w)$.

Let us now give the intuition behind the ESO inequality \eqref{eq:ESO}. Assuming the current iterate is $x$, PCDM changes $x^{(i)}$ to $x^{(i)}+h^{(i)}(x)$ for $i \in \hat{S}$, where $h(x)$ is the minimizer of the right hand side of \eqref{eq:ESO}.
By doing so, we benefit from the following:
\begin{enumerate}
\item[(i)] Since the overapproximation is a convex quadratic in $h$, \emph{it easy to compute $h(x)$}.
\item[(ii)] Since the overapproximation is  block separable, one \emph{can compute the updates $h^{(i)}(x)$ in parallel} for all $i\in\{1,2,\dots,n\}$.
\item[(iii)] For the same reason, one \emph{can compute the updates  or $i \in S_k$ only}, where $S_k$ is the sample set drawn at iteration $k$ following the law describing $\hat{S}$.
\end{enumerate}
The algorithmic strategy of PCDM is to move to a new point in such a way that the expected value of the loss function evaluated at this new point is as small as possible. The method effectively decomposes the $N$-dimensional problem into $n$ smaller convex quadratic problems, attending to a random subset of $\tau$ of them at each iteration, in parallel. A single iteration of PCDM  can  be  compactly written as
\begin{equation}\label{eq:PCDM-smooth-case}x \leftarrow x+ (h(x))_{[\hat{S}]},\end{equation} where $h(x) = (h^{(1)}(x),\dots,h^{(n)}(x))$ and
\begin{equation}\label{eq:stepsize}h^{(i)}(x) = \arg \min_h \left\{\ve{(\nabla \phi(x))^{(i)}}{h^{(i)}} + \frac{\beta w_i}{2}\ve{B_i h^{(i)}}{h^{(i)}}\right\} \overset{\eqref{eq:ESO}}{=}  -\frac{1}{\beta w_i} B_i^{-1} (\nabla \phi(x))^{(i)}.\end{equation}

From the update formula \eqref{eq:stepsize} we can see that $\tfrac{1}{\beta}$ can be interpreted as a stepsize. We would hence wish to choose  small $\beta$, but not too small so that the method does not diverge. The issue of the computation of a good (small) parameter $\beta$ is very intricate for several reasons, and is at the heart of the design of a randomized parallel coordinate descent method. Much of the theory developed in this paper is aimed at identifying a class of nonsmooth composite problems which, when smoothed, admit ESO with a small and easily computable value of $\beta$. In the following text we give some insight into why this issue is difficult, still in the simplified smooth setting.

\subsection{Spurious ways of computing $\beta$} \label{sec:spurious}

Recall that the parameters $\beta$ and $w$ giving rise to an ESO need to be \emph{explicitly calculated} before the method is run as they are needed in the computation of the update steps. We will now describe the issues associated with finding suitable $\beta$, for simplicity assuming that $w$ has been chosen/computed.
\begin{enumerate}
\item Let us start with a first approach to computing $\beta$. If the gradient of $ \phi$ is Lipschitz with respect to the separable norm
\[\|x\|_w^2 \eqdef \sum_{i=1}^n w_i \ve{B_i x^{(i)}}{x^{(i)}},\] with \emph{known} Lipschitz constant $L$, then for all $x, h\in \R^N$ we have
$\phi(x+h') \leq \phi(x) + \ve{\nabla \phi(x)}{h'} +\frac{L}{2}\|h'\|_w^2$. Now, if for fixed $h \in \R^N$ we substitute $h'=h_{[\hat{S}]}$ into this inequality, and take expectations utilizing the identities~\cite{RT:PCDM}
\begin{equation}
\label{eq:jss8s8s}\Exp\left[\ve{x}{h_{[\hat{S}]}}\right] = \frac{\Exp[|\hat{S}|]}{n}\ve{x}{h}, \qquad \Exp\left[\|h_{[\hat{S}]}\|_w^2\right] = \frac{\Exp[|\hat{S}|]}{n}\|h\|_w^2,
\end{equation}

we obtain $(\phi,\hat{S})\sim \ESO(\beta,w)$ for $\beta=L$. It turns out that this way of obtaining $\beta$ is far from satisfactory, for several reasons.
\begin{enumerate}
\item First, it is very difficult to compute $L$ in the big data setting PCDMs are designed for. In the case of L2 regression, for instance, $L$ will be equal to the largest eigenvalue of a certain $N\times N$ matrix. For huge $N$, this is a formidable task, and may actually be harder than the problem we are trying to solve.
\item We show in Section~\ref{sec:ESO+Lipschitz} that taking $\beta = \tfrac{n}{\tau}c$, where $c$ is a  bound on the Lipschitz constants (with respect to the norm $\|\cdot\|_w$, at $h=0$, uniform in $x$) of the gradients of the functions $h \to \Exp[\phi(x+h_{[\hat{S}]})]$ precisely characterizes \eqref{eq:ESO}, and leads to smaller (=better) values $\beta$. Surprisingly, this $\beta$ can be $O(\sqrt{n})$ \emph{times} smaller than $L$. As we shall see, this directly translates into iteration complexity speedup by the factor of $O(\sqrt{n})$.
\end{enumerate}

\item It is often easy to obtain good $\beta$ in the case $\tau=1$. Indeed, it follows from \cite{Nesterov:2010RCDM, RT:UCDC} that any smooth convex function $\phi$ will satisfy \eqref{eq:ESO} with $\beta=1$ and $w_i=L_i$, where $L_i$ is the block Lipschitz constant of the gradient of $\phi$ with respect to the norm $\ve{B_i \cdot}{\cdot}^{1/2}$, associated with block $i$. If the size of  block $i$ is $N_i$, then the computation of $L_i$ will typically amount to the finding a maximal eigenvalue of an $N_i\times N_i$ matrix. If the block sizes $N_i$ are sufficiently small, it is much simpler to compute $n$ of these quantities than to compute $L$.  Now, can we use a similar technique to obtain $\beta$ in the $\tau>1$ case?  A naive idea would  be to keep $\beta$ unchanged ($\beta=1$). In view of \eqref{eq:stepsize}, this means that one would simply compute the updates $h^{(i)}(x)$ in the same way as in the $\tau=1$ case, and apply them all. However, this strategy is doomed to fail:  the method may end up oscillating between  sub-optimal points (a simple 2 dimensional example was described in \cite{minibatch-ICML2013}). This issue arises since the algorithm overshoots: while the individual updates are safe for $\tau=1$, it is not clear why adding them all up for arbitrary $\tau$ should  decrease the function value.

\item A natural remedy to the problem described in $\S2$ is to decrease the stepsize, i.e., to increase $\beta$ as $\tau$ increases. In fact, it can be inferred from \cite{RT:PCDM} that $\beta(\tau) = \tau$ always works: it satisfies the ESO inequality and the method converges. This makes intuitive sense since the actual step in the $\tau>1$ case is obtained as the \emph{average} of the block updates which are safe in the $\tau=1$ case. By Jensen's inequality, this must decrease the objective function since the randomized serial method does (below we assume for notational simplicity that all blocks are of size one, $e_i$ are the unit coordinate vectors):
\[\phi(x_{+}) = \phi\left(x - \sum_{i \in \hat{S}}  \tfrac{1}{\tau L_i} (\nabla \phi(x))^{(i)} e_i \right) \leq \tfrac{1}{\tau}\sum_{i\in \hat{S}} \phi\left(x - \tfrac{1}{L_i}(\nabla \phi(x))^{(i)}e_i\right).\]

    However, this approach compensates the increase of computational power ($\tau$) by the same decrease in stepsize, which means that the parallel method ($\tau>1$) might in the worst case require the same number of iterations as the serial one ($\tau=1$).

\item The issues described in $\S2$ and $\S3$ lead us to the following question: Is it possible to \emph{safely} and \emph{quickly} choose/compute a value of $\beta$ in the $\tau=1$ case which is larger than $1$ but smaller than $\tau$? If this was possible, we could expect the parallel method to be much better than its serial counterpart. An affirmative answer to this question for the class of smooth convex partially separable functions $\phi$ was given in \cite{RT:PCDM}.
\end{enumerate}

To summarize, the issue of selecting $\beta$ in the parallel setting is very intricate, and of utmost significance for the algorithm.  In the next two subsections we now give more insight into this issue and in doing so progress into discussing our contributions.

\subsection{Generic complexity results and partial separability} \label{sec:generic}

The generic complexity results mentioned earlier, established in \cite{RT:PCDM} for PCDM, have the form\footnote{This holds provided $w$ does not change with $\tau$; which is the case in this paper and in the smooth partially separable setting considered in \cite[Section~6]{RT:PCDM}. Also, for simplicity we cast the results here in the case $\Psi\equiv 0$, but they hold in the composite case as well.}
\[k \geq \left(\frac{\beta}{\tau}\right) \times n \times c \qquad \Rightarrow \qquad \Prob\left(\phi(x_k) - \min_x \phi(x) \leq \epsilon\right) \geq 1-\rho,\]
where $c$ is a constant independent of $\tau$, and depending  on the error tolerance $\epsilon$,
confidence tolerance $\rho$, initial iterate $x_0$, optimal point $x^*$ and $w$. Moreover, $c$ does not hide any large constants.

Keeping $\tau$ fixed, from  \eqref{eq:stepsize} we see that larger values of $\beta$ lead to smaller stepsizes. We commented earlier, appealing to intuition, that this translates into worse complexity. This is now affirmed and quantified by the above generic complexity result. Note, however,  that this generic result \emph{does not provide any concrete information about parallelization speedup} because it does not say anything about the dependence of $\beta$ on $\tau$. Clearly, parallelization speedup occurs when the function \[T(\tau) = \frac{\beta(\tau)}{\tau}\]
is decreasing. The behavior of this function is important for big data problems which can only be solved by decomposition methods, such as PCDM, on modern HPC architectures.


Besides proving generic complexity bounds for PCDM, as outlined above,  Richt\'{a}rik and Tak\'{a}\v{c} \cite{RT:PCDM} identified a class of \emph{smooth convex} functions $\phi$ for which $\beta$ can be explicitly computed as a function of $\tau$ in closed form, and for which indeed $T(\tau)$ is decreasing: \emph{partially separable} functions. A convex function $\phi : \R^N \to \R$ is partially separable of degree $\omega$ if it can be written as a sum of differentiable\footnote{It is not assumed that the summands have Lipschitz gradient.} convex functions, each of which depends on at most $\omega$ of the $n$ blocks of $x$. If $\hat{S}$ is a   $\tau$-uniform sampling,
  then $\beta = \beta' = \min\{\omega,\tau\}$. If $\hat{S}$ is a $\tau$-nice sampling, then $\beta = \beta'' = 1 + \tfrac{(\omega-1)(\tau-1)}{n-1}$. Note that $\beta''\leq \beta'$ and that $\beta'$ can be arbitrarily larger than $\beta''$. Indeed, the worst case situation (in terms of the ratio $\tfrac{\beta'}{\beta''}$) for any fixed $n$ is $\omega = \tau = \sqrt{n}$, in which case
\[\frac{\beta'}{\beta''} = \frac{1+\sqrt{n}}{2}.\]

This means that PCDM implemented with a $\tau$-nice sampling (using $\beta''$) can be arbitrarily faster than PCDM implemented with the more general $\tau$-uniform sampling (using $\beta'$). This simple example illustrates the huge impact the choice of the sampling $\hat{S}$ has, other things equal. As we shall show in this paper, this phenomenon is directly related to the issue we discussed in Section~\ref{sec:spurious}: $L$ can be $O(\sqrt{n})$ times larger than a good $\beta$.

\subsection{Brief literature review}



\textbf{Serial randomized methods.} Leventhal and Lewis \cite{Leventhal:2008:RMLC} studied the complexity of randomized coordinate descent methods for the minimization of convex quadratics and proved that the method converges linearly even in the non-strongly convex case. Linear convergence for smooth strongly convex functions was proved by Nesterov~\cite{Nesterov:2010RCDM} and for general regularized problems by Richt\'{a}rik and Tak\'{a}\v{c}~\cite{RT:UCDC}. Complexity results for smooth problems with special regularizes (box constraints, L1 norm) were obtained by Shalev-Shwarz and Tewari~\cite{ShalevTewari09} and Nesterov~\cite{Nesterov:2010RCDM}. Nesterov was the first to analyze the block setting, and proposed using different Lipschitz constants for different blocks, which has a big impact on the efficiency of the method since these constants capture important second order information \cite{Nesterov:2010RCDM}. Also, he was the first to analyze an accelerated coordinate descent method. Richt\'{a}rik and Tak\'{a}\v{c}~\cite{RT:SPARS11, RT:UCDC} improved, generalized and simplified previous results and extended the analysis to the composite case. They also gave the first analysis of a coordinate descent method using arbitrary probabilities. Lu and Xiao~\cite{luxiao2013complexity} recently studied the work developed in \cite{Nesterov:2010RCDM} and \cite{RT:PCDM} and obtained further improvements. Coordinate descent methods were recently extended to deal with coupled constraints by Necoara et al \cite{Necoara:Coupled} and extended to the composite setting by Necoara and Patrascu~\cite{Necoara_composite-coupled}. When the function is not smooth neither composite, it is still possible to define coordinate descent methods with subgradients. An algorithm based on the averaging of past subgradient coordinates is presented in~\cite{tao2012stochastic} and a successful subgradient-based coordinate descent method for problems with sparse subgradients is proposed by Nesterov~\cite{Nesterov-Subgrad-Huge}. Tappenden et al~\cite{TRG:InexactCDM}  analyzed an inexact randomized coordinate descent method in which proximal subproblems at each iteration are solved only approximately. Dang and Lan~\cite{LanBCD2013} studied complexity of stochastic block mirror descent methods for nonsmooth and stochastic optimization and an accelerated method was studies by Shalev-Shwarz and Zhang~\cite{SSS2013-accelerated}. Lacoste-Julien et al~\cite{Jaggi-ICML2013-block-Frank-Wolfe} were the first to develop a block-coordinate Frank-Wolfe method. The generalized power method of Journ\'{e}e et al \cite{Richtarik-GPower2010} designed for sparse PCA can be seen as a nonconvex block coordinate ascent method with two blocks \cite{RTA:24am}.

\textbf{Parallel methods.} One of the first complexity results for a parallel coordinate descent method  was obtained by Ruszczy\'{n}ski~\cite{Ruszczynski95} and is known as the diagonal quadratic approximation method (DQAM). DQAM updates all blocks at each iteration, and hence is not randomized. The method was designed for solving a convex composite problem with quadratic smooth part and arbitrary separable nonsmooth part and was motivated by the need to solve separable linearly constrained problems arising in stochastic programming. As described in previous sections, a family of randomized parallel block coordinate descent methods (PCDM) for convex composite problems was analyzed by Richt\'{a}rik and Tak\'{a}\v{c}~\cite{RT:PCDM}. Tappenden et al \cite{TRB2013:DQA} recently contrasted the DQA method~\cite{Ruszczynski95} with PCDM \cite{RT:PCDM}, improved the complexity result \cite{RT:PCDM} in the strongly convex case and showed that for PCDM it is optimal choose $\tau$ to be equal to the number of processors. Utilizing the ESO machinery \cite{RT:PCDM} and the primal-dual technique developed by Shalev-Shwarz and Zhang~\cite{Stoch-dual-Coord-Ascent}, Tak\'{a}\v{c} et al \cite{minibatch-ICML2013} developed and analyzed a parallel (mini-batch) stochastic subgradient descent method (applied to the primal problem of training support vector machines with the hinge loss) and a parallel  stochastic dual coordinate ascent method (applied to the dual box-constrained concave maximization problem). The analysis naturally extends to the general setting of Shalev-Shwarz and Zhang~\cite{Stoch-dual-Coord-Ascent}. A parallel Newton coordinate descent method was proposed in \cite{PCDN_2013}. Parallel methods for L1 regularized problems with an application to truss topology design were proposed  by Richt\'{a}rik and Tak\'{a}\v{c}~\cite{RT:TTD2011}. They give the first analysis of a greedy serial coordinate descent method for L1 regularized problems. An early analysis of a PCDM  for L1 regularized problems was performed by Bradley et al \cite{Bradley:PCD-paper}. Other recent parallel methods include \cite{Necoara:parallelCDM-MPC, BOOM2013}.

\subsection{Contents}

In Section~\ref{sec:SPCDM} we describe the problems we study, the algorithm (smoothed parallel coordinate descent method), review Nesterov's smoothing technique and enumerate  our contributions. In Section~\ref{sec:DSO} we compute Lipschitz constants of the gradient smooth approximations of Nesterov separable functions associated with subspaces spanned by arbitrary subset of blocks, and in Section~\ref{sec:ESO} we derive ESO inequalities. Complexity results are derived in Section~\ref{sec:complexity} and finally, in Section~\ref{sec:applications} we describe three applications and preliminary numerical experiments.

\section{Smoothed Parallel Coordinate Descent Method} \label{sec:SPCDM}

In this section we describe the problems we study, the algorithm and list our contributions.

\subsection{Nonsmooth and smoothed composite problems}\label{Sec:8sjs8s}

In this paper we study the iteration complexity of PCDMs applied to two classes of convex composite optimization problems:

\begin{equation}
\label{eq:P}
 \text{minimize} \quad  \FF(x)\eqdef \ff(x)+ \OO(x) \quad \text{subject to} \quad x\in \R^N,
\end{equation}
and
\begin{equation}
\label{eq:Psmooth}
 \text{minimize} \quad  \FF_\mu(x)\eqdef \ff_\mu(x)+ \OO(x) \quad \text{subject to} \quad x\in \R^N.
\end{equation}
We assume \eqref{eq:P} has an optimal solution ($x^*$) and consider the following setup:


\begin{enumerate}
\item \textbf{(Structure of $\ff$)} First, we assume that $\ff$ is of the form
\begin{equation}\label{eq:defstructured}
\ff(x) \eqdef \max_{\z\in \Q} \{ \ve{A x}{\z} - \hphi(\z)\},
\end{equation}
where  $\Q\subseteq \R^m$ is a nonempty compact convex set, $A\in \R^{m\times N}$, $\hphi: \R^m\to \R$ is convex and $\ve{\cdot}{\cdot}$ is the standard Euclidean inner product (the sum of products of the coordinates of the vectors). Note that $\ff$ is convex and in general \emph{nonsmooth}.

\item \textbf{(Structure of $\ff_\mu$)} Further, we assume that $\ff_\mu$ is of the form
\begin{equation} \label{eq:f_mu-P2}\ff_\mu(x) \eqdef \max_{\z \in \Q} \{\ve{Ax}{\z} - \hphi(\z) - \mu d(\z)\},\end{equation}
where $A, \Q$ and $g$ are as above, $\mu>0$ and $d:\R^m\to \R$ is $\sigma$-strongly convex on $\Q$ with respect to the norm
\begin{equation}\label{eq:p-norm}\|\z\|_v \eqdef \Big(\sum_{j=1}^m v_j^p |\z_j|^{p}\Big)^{1/p},\end{equation}
where $v_1,\dots, v_m$ are positive scalars, $1\leq p \leq 2$ and $\z = (\z_1,\dots,\z_m)^T \in \R^m$. We further assume that $d$ is nonnegative on $\Q$  and that $d(\z_0)=0$ for some $\z_0\in \Q$. It then follows that $d(\z) \geq \tfrac{\sigma}{2}\|\z-z_0\|_v^2$ for all $\z \in Q$. That is, $d$ is a prox function on $Q$. We further let $D \eqdef \max_{z \in Q} d(z)$.

For $p> 1$ let $q$ be such that $\tfrac{1}{p}+\tfrac{1}{q} = 1$. Then the conjugate norm of $\|\cdot\|_v$ defined in \eqref{eq:p-norm} is given by
\begin{equation}\label{eq:q-norm}\|\z\|_v^* \eqdef \max_{\|\z'\|_v\leq 1} \ve{\z'}{\z} = \begin{cases}\left(\sum_{j=1}^m v_j^{-q} |\z_j|^{q}\right)^{1/q}, & 1<p\leq 2,\\
\max_{1\leq j \leq m} v_j^{-1} |\z_j|, & p=1.\end{cases}\end{equation}

It is well known that $\ff_\mu$ is a \emph{smooth} convex function; i.e., it is differentiable and its  gradient is Lipschitz.

{\footnotesize \emph{Remark:} As shown by Nesterov in his seminal work on smooth minimization of nonsmooth functions \cite{Nesterov05:smooth}---here summarized in Proposition~\ref{prop:Nesterov}---$\ff_\mu$ is a smooth approximation of $\ff$. In this paper, when solving \eqref{eq:P}, we apply PCDM to \eqref{eq:Psmooth} for a specific choice of $\mu>0$, and then argue, following now-standard reasoning from \cite{Nesterov05:smooth}, that the solution is an approximate solution of the original problem.  This will be made precise in Section~\ref{sec:smoothing}. However, in some cases one is interested in minimizing a function of the form \eqref{eq:Psmooth} directly, without the need to  interpret $\ff_\mu$ as a smooth approximation of another function. For instance, as we shall see in Section~\ref{sec:adaboost}, this is the case with the ``AdaBoost problem''. In summary, both problems \eqref{eq:P} and \eqref{eq:Psmooth} are of interest on their own, even though our approach to solving the first one is by transforming it to the second one.
}

\item \textbf{(Block structure)} Let $A = [A_1,A_2,\dots, A_n]$ be decomposed into nonzero column submatrices, where $A_i \in \R^{m\times N_i}$, $N_i\geq 1$ and $\sum_{i=1}^n N_i = N$, and $U=[U_1,U_2,\dots,U_n]$ be a decomposition of the $N \times N$ identity matrix $U$ into  submatrices $U_i \in \R^{N \times N_i}$. Note that \begin{equation}\label{eq:A_i}A_i = AU_i.\end{equation} It will be useful to note that
\begin{equation}\label{eq:U_iU_j}U_i^T U_j = \begin{cases} N_i \times N_i \text{ identity matrix}, & i = j,\\
N_i \times N_j \text{ zero matrix}, & \text{otherwise}.
\end{cases}
\end{equation}

     For $x\in \R^N$, let $x^{(i)}$ be the block of variables corresponding to the columns of $A$ captured by $A_i$, that is, $x^{(i)} = U_i^T x \in \R^{N_i}$, $i=1,2,\dots,n$. Clearly, any vector $x \in \R^N$ can be written uniquely as $x = \sum_{i=1}^n U_i x^{(i)}$. We will often refer to the vector $x^{(i)}$ as the \emph{i-th block} of $x$. We can now formalize the notation used in the introduction (e.g., in \eqref{eq:PCDM-smooth-case}):
     for $h \in \R^N$ and $\emptyset \neq S\subseteq [n]\eqdef \{1,2,\dots,n\}$ it will be convenient  to write \begin{equation}\label{h_[S]}h_{[S]} \eqdef \sum_{i\in S} U_i h^{(i)}.\end{equation}
     Finally, with each block $i$ we associate a positive definite matrix $B_i \in \R^{N_i \times N_i}$ and scalar $w_i>0$, and equip $\R^N$ with a pair of conjugate norms:
\begin{equation}\label{eq:norm_block}\|x\|_w^2  \eqdef \sum_{i=1}^n w_i \ve{B_i x^{(i)}}{x^{(i)}}, \qquad (\|y\|_w^*)^2  \eqdef \max_{\|x\|_w \leq 1}\ve{y}{x}^2 = \sum_{i=1}^n w_i^{-1} \ve{B_i^{-1}y^{(i)}}{y^{(i)}}.
\end{equation}

{\footnotesize \emph{Remark:} For some problems, it is relevant to consider blocks of coordinates as opposed to individual coordinates. The novel aspects of this paper are \emph{not} in the block setup however, which was already considered in \cite{Nesterov:2010RCDM, RT:PCDM}. We still write the paper in the general block setting; for several reasons. First, it is often practical to work with blocks either due to the nature of the problem (e.g., group lasso), or due to numerical considerations (it is often more efficient to process a ``block'' of coordinates at the same time). Moreover, some parts of the theory need to be treated differently in the block setting. The theory, however, does not get more complicated due to the introduction of blocks. A small notational overhead is a small price to pay for these benefits.
}

\item \textbf{(Sparsity of $A$)} For a vector $x\in \R^N$ let \begin{equation}\label{eq:jd8307d}\J(x) \eqdef \{i\st U_i^T x \neq 0\} =  \{i\st x^{(i)} \neq 0\}.\end{equation} Let $A_{ji}$ be the $j$-th row of $A_i$. If $e_1,\dots,e_m$ are the unit coordinate vectors in $\R^m$, then
\begin{equation}\label{eq:A_ji}A_{ji} \eqdef e_j^T A_i.\end{equation}
Using the above notation, the set of nonzero blocks of the $j$-th row of $A$ can   be expressed as
\begin{equation}
\label{eq:subvector}\J(A^T e_j) \overset{\eqref{eq:jd8307d}}{=} \{i\st U_i^T A^T e_j \neq 0\} \overset{\eqref{eq:A_i} + \eqref{eq:A_ji}}{=} \{i\st A_{ji}\neq 0\}.
\end{equation}

The following concept is key to this paper.

\begin{defn}[Nesterov separability\footnote{We coined the term \emph{Nesterov separability} in honor of Yu.\ Nesterov's seminal work on the smoothing technique \cite{Nesterov05:smooth}, which is applicable to functions represented in the form \eqref{eq:defstructured}.  Nesterov did not study problems
with row-sparse matrices $A$, as we do in this work, nor did he study parallel coordinate descent methods. However, he proposed  the celebrated smoothing technique which we also  employ in this paper.
}]
\label{defn:maxpartsep}
We say that $\ff$  (resp. $\ff_\mu$) is \emph{Nesterov (block) separable} of degree $\omega$ if it has the form \eqref{eq:defstructured} (resp. \eqref{eq:f_mu-P2}) and
\begin{equation}\label{eq:omega}
\max_{1\leq j \leq m} |\J(A^T e_j)| \leq \omega.
\end{equation}
\end{defn}

Note that in the special case when all blocks are of cardinality 1 (i.e., $N_i=1$ for all $i$),  the above definition simply requires all rows of $A$ to have at most $\omega$ nonzero entries.

\item \textbf{(Separability of $\OO$)} We assume that \[\OO(x) = \sum_{i=1}^n \OO_i(x^{(i)}),\] where $\OO_i: \R^{N_i} \to \R \cup \{+ \infty\}$ are \emph{simple} proper closed convex functions.

{\footnotesize \emph{Remark:} Note that we do not assume that the functions $\OO_i$ be smooth. In fact, the most interesting cases in terms of applications are nonsmooth functions such as, for instance, i) $\OO_i(t) = \lambda |t|$ for some $\lambda>0$ and all $i$ (L1 regularized optimization), ii) $\OO_i(t) = 0$ for $t \in [a_i,b_i]$, where $-\infty \leq a_i \leq b_i \leq +\infty$ are some constants, and $\OO_i(t) = +\infty$ for $t \notin [a_i,b_i]$ (box constrained optimization).

}
\end{enumerate}


We are now ready to state  the method (Algorithm~\ref{alg:pcdm1}) we use for solving the smoothed composite problem \eqref{eq:Psmooth}. Note that for $\phi\equiv \ff_\mu$ and $\Psi\equiv 0$, Algorithm~\ref{alg:pcdm1} coincides with the  method \eqref{eq:PCDM-smooth-case}-\eqref{eq:stepsize} described in the introduction. The only conceptual difference here is that in the computation of the updates in Step 2 we need to augment the quadratic obtained from ESO with $\Psi$.  Note that Step 3 can be compactly written as
\begin{equation}\label{eq:step} x_{k+1} = x_k + (h_k)_{[S_k]}.\end{equation}

\begin{algorithm}
  \begin{algorithmic}
\STATE \textbf{Input:} initial iterate $x_0\in \R^N$, $\beta > 0$ and $w = (w_1,\dots,w_n) > 0$
\FOR{$k \geq 0$}
  \STATE \textbf{Step 1.} Generate a random set of blocks $S_k \subseteq \{1,2,\dots,n\}$ 
  \STATE \textbf{Step 2.} In parallel for $i \in S_k$, compute \[h_k^{(i)} = \arg \min_{t\in \R^{N_i}} \left\{\ve{(\nabla \ff_{\mu}(x_k))^{(i)}}{t} + \frac{\beta w_i}{2} \ve{B_i t}{t} + \OO_i(x_k^{(i)}+t)\right\} \]
  \STATE \textbf{Step 3.} In parallel for $i \in S_k$, update  $x_{k}^{(i)} \leftarrow x_k^{(i)} + h_k^{(i)}$ and set $x_{k+1} \leftarrow x_k$
\ENDFOR
  \end{algorithmic}
\caption{Smoothed Parallel Coordinate Descent Method (SPCDM)}
\label{alg:pcdm1}
\end{algorithm}

 Let us remark that the scheme actually encodes an entire family of methods. For $\tau=1$ we have a serial method (one block updated per iteration), for $\tau=n$ we have a fully parallel method (all blocks updated in each iteration), and there are many partially parallel methods in between, depending on the choice of $\tau$. Likewise, there is flexibility in choosing the block structure. For instance, if we choose $N_i=1$ for all $i$, we have a proximal coordinate descent method, for $N_i>1$, we have a proximal block coordinate descent and for $n=1$ we have a proximal gradient descent method.



\subsection{Nesterov's smoothing technique}\label{sec:smoothing}

In the rest of the paper we will repeatedly make use of the now-classical smoothing technique of Nesterov \cite{Nesterov05:smooth}. We will \emph{not}
use this merely to approximate $\ff$ by $\ff_\mu$; the technique will be utilized in several proofs in other ways, too. In this section we collect the facts that we will need.

Let $\E_1$ and $\E_2$ be two finite dimensional linear normed spaces, and $\E_1^*$ and $\E_2^*$ be their duals (i.e., the spaces of bounded linear functionals). We equip $\E_1$ and $\E_2$ with norms $\|\cdot\|_1$ and $\|\cdot\|_2$, and the dual spaces $\E_1^*$, $\E_2^*$ with the dual (conjugate norms): \[\|y\|_j^* \eqdef \max_{\|x\|_j \leq 1} \ve{y}{x}, \qquad y \in \E_j^*, \qquad j=1,2,\]
where $\ve{y}{x}$ denotes the action of the linear functional $y$ on $x$. Let $\bar{A} : \E_1 \to \E_2^*$ be a linear operator, and let $\bar{A}^*: \E_2 \to \E_1^*$ be its adjoint:
\[\ve{\bar{A}x}{u} = \ve{x}{\bar{A}^*u}, \qquad x \in \E_1, \qquad u \in \E_2.\]

Let us equip $\bar{A}$ with a norm as follows:
\begin{eqnarray}\label{eq:12norms}
\|\bar{A}\|_{1,2} & \eqdef & \max_{x,u} \left\{\ve{Ax}{u} \st x \in \E_1,\; \|x\|_1 = 1, \; u \in \E_2,\; \|u\|_2 = 1 \right\}\notag\\
& = & \max_{x} \{\|\bar{A}x\|_2^* \st x \in \E_1, \; \|x\|_1 = 1\} = \max_u \{\|\bar{A}^* u\|_1^* \st u\in \E_2,\; \|u\|_2 = 1\}.
\end{eqnarray}

Consider now the function $\bar{\ff} : \E_1 \to \R$ given by
\[\bar{\ff}(x) = \max_{u \in \bar{Q}} \{ \ve{\bar{A}x}{u} - \bar{g}(u) \},\]
where $\bar{Q} \subset \E_2$ is a compact convex set and $\bar{g} : \E_2 \to \R$ is convex. Clearly, $\bar{\ff}$ is convex and in general nonsmooth.

We now describe Nesterov's smoothing technique for approximating $\bar{\ff}$ by a convex function with Lipschitz gradient. The technique relies on the introduction of a prox-function $\bar{d}: \E_2 \to \R$. This function is continuous and strongly convex on $\bar{Q}$ with convexity
parameter $\bar{\sigma}$. Let $u_0$ be the minimizer of $\bar{d}$ on $\bar{Q}$. Without loss of generality, we can assume that
$\bar{d}(u_0)=0$ so that for all $u \in \bar{Q}$, $\bar{d}(u)\geq \frac{\bar{\sigma}}{2}\|u-u_0\|_2^2$. We also write $\bar{D} \eqdef \max \{\bar{d}(u) \st u \in \bar{Q}\}$. Nesterov's smooth approximation of  $\bar{\ff}$ is defined for any $\mu>0$ by
\begin{equation} \label{eq:f_mu}\bar{\ff}_\mu(x) \eqdef \max_{u \in \bar{Q}} \{\ve{\bar{A}x}{u} - \bar{g}(u) - \mu \bar{d}(u)\}.\end{equation}

\begin{prop}[Nesterov \cite{Nesterov05:smooth}] \label{prop:Nesterov} The function $\bar{\ff}_\mu$ is continuously differentiable on $\E_1$ and satisfies
\begin{equation}\label{eq:eps09809}\bar{\ff}_\mu(x) \leq \bar{\ff}(x) \leq \bar{\ff}_\mu(x) + \mu \bar{D}.\end{equation}
Moreover, $\bar{\ff}_\mu$ is convex and its gradient $\nabla \bar{\ff}_\mu(x) = \bar{A}^* u^*$, where $u^*$ is the unique maximizer in \eqref{eq:f_mu},
is Lipschitz continuous with constant
\begin{equation}\label{eq:L_mu}L_\mu = \frac{1}{\mu \bar{\sigma}}\|\bar{A}\|^2_{1,2}.\end{equation}
That is, for all $x,h\in \E_1$,
\begin{equation}\label{eq:DSO-Nesterov}\bar{\ff}_\mu(x+h) \leq \bar{\ff}_\mu(x) + \langle \nabla \bar{\ff}_\mu(x), h \rangle + \frac{\|\bar{A}\|^2_{1,2}}{2 \mu \bar{\sigma}} \|h\|_1^2.\end{equation}

\end{prop}

The above result will be used in this paper in various ways:

\begin{enumerate}
\item As a direct consequence of \eqref{eq:DSO-Nesterov} for $\E_1 = \R^N$ (primal basic space), $\E_2=\R^m$ (dual basic space), $\|\cdot\|_1 = \|\cdot\|_w$, $\|\cdot\|_2=\|\cdot\|_v$, $\bar{d}=d$, $\bar{\sigma}=\sigma$, $\bar{Q}=Q$, $\bar{g}=g$, $\bar{A}=A$ and $\bar{\ff}=\ff$,
we obtain the following inequality:
\begin{equation}\label{eq:DSO-Nesterov-corol}\ff_\mu(x+h) \leq \ff_\mu(x) + \langle \nabla \ff_\mu(x), h \rangle + \frac{\|A\|^2_{w,v}}{2 \mu \sigma} \|h\|_w^2.\end{equation}

\item A large part of this paper is devoted to various refinements (for a carefully chosen data-dependent $w$ we ``replace'' $\|A\|^2_{w,v}$ by an easily computable and interpretable quantity depending on $h$ and $\omega$, which gets smaller as $h$ gets sparser and $\omega$ decreases) and extensions (left-hand side is replaced by $\Exp [f_\mu(x+h_{[\hat{S}]})]$) of inequality \eqref{eq:DSO-Nesterov-corol}. In particular, we give formulas for  fast computation of subspace Lipschitz constants of $\nabla \ff_\mu$ (Section~\ref{sec:DSO}) and derive ESO inequalities (Section~\ref{sec:ESO})---which are essential for proving iteration complexity results for variants of the smoothed parallel coordinate descent method (Algorithm~\ref{alg:pcdm1}).


\item Besides the above application to smoothing $\ff$; we will utilize Proposition~\ref{prop:Nesterov} also as a tool for computing Lipschitz constants of the gradient of two technical functions needed in proofs. In Section~\ref{sec:DSO} we will use $\E_1=\R^S$ (``primal update space'' associated with a subset $S\subseteq [n]$), $\E_2=\R^m$ and $\bar{A}=A^{(S)}$. In Section~\ref{sec:ESO} we will use $\E_1 = \R^N$, $\E_2=\R^{|{\cal P}| \times m}$ (``dual product space'' associated with sampling $\hat{S}$) and $\bar{A}=\hat{A}$. These spaces and matrices will be defined in the above mentioned sections, where they are needed.

\end{enumerate}

The following simple consequence of Proposition~\ref{prop:Nesterov} will be useful in proving our complexity results.

\begin{lem}\label{eq:lemma098098sdsd} Let $x^*$ be an optimal solution of \eqref{eq:P} (i.e., $x^*=\arg \min_x \FF(x)$) and $x_\mu^*$ be an optimal solution of \eqref{eq:Psmooth} (i.e., $x_\mu^* = \arg \min_{x} F_\mu(x)$). Then for any $x \in \dom \Psi$ and $\mu>0$,
\begin{equation}\label{eq:propsjhs88sd}\FF_\mu(x) - \FF_\mu(x_\mu^*) - \mu D \quad \leq \quad \FF(x) - \FF(x^*) \quad \leq \quad \FF_\mu(x) - \FF_\mu(x_\mu^*) + \mu D.\end{equation}
\end{lem}
\begin{proof} From Proposition~\ref{prop:Nesterov} (used with $\bar{A}=A$, $\bar{\ff} = \ff$, $\bar{Q}=Q$, $\bar{d} = d$,  $\|\cdot\|_2 = \|\cdot\|_v$, $\bar{\sigma}=\sigma$, $\bar{D}=D$  and $\bar{\ff}_\mu = \ff_\mu$), we get $\ff_\mu(y) \leq \ff(y) \leq \ff_\mu(y) +\mu D$, and adding $\Psi(y)$ to all terms leads to $\FF_\mu(y) \leq \FF(y) \leq \FF_\mu(y) +\mu D$, for all $y \in \dom \Psi$. We only prove the second inequality, the first one can be shown analogously. From the last chain of inequalities and optimality of $x_\mu^*$ we get i) $\FF(x) \leq \FF_\mu(x)+\mu D$ and ii) $\FF_\mu(x_\mu^*)\leq \FF_\mu(x^*) \leq \FF(x^*)$. We only need to subtract (ii) from (i).
\end{proof}

\subsection{Contributions}

We now describe some of the main contributions of this work.

\begin{enumerate}

\item \textbf{First complexity results.} We give the first complexity results for solving problems \eqref{eq:P} and \eqref{eq:Psmooth} by a parallel coordinate descent method. In fact, to the best of our knowledge, we are not aware of any complexity results even in the $\Psi\equiv 0$ case.
    We obtain our results by combining the following: i) we show that $\ff_\mu$---smooth approximation of $\ff$---admits ESO inequalities with respect to uniform samplings and compute ``good'' parameters $\beta$ and $w$, ii) for problem \eqref{eq:P} we utilize Nesterov's smoothing results (via Lemma~\eqref{eq:lemma098098sdsd}) to argue that an approximate solution of  \eqref{eq:Psmooth} is an approximate solution of \eqref{eq:P}, iii) we use the generic complexity bounds proved by Richt\'{a}rik and Tak\'{a}\v{c} \cite{RT:PCDM}.

\item \textbf{Nesterov separability.} We identify the degree  of \emph{Nesterov separability} as the important quantity driving parallelization speedup.

\item \textbf{ESO parameters.} We show that it is possible to compute ESO parameters $\beta$ and $w$ \emph{easily}. This is of utmost importance for big data applications where the computation of the Lipschitz constant $L$ of $\nabla \phi= \nabla \ff_\mu$ is prohibitively expensive (recall the discussion in Section~\ref{sec:spurious}). In particular, we suggest that in the case with all blocks being of size 1 ($N_i=1$ and $B_i=1$ for all $i$), the weights $w_i=w_i^*$, $i=1,2,\dots,n$, be chosen as follows:
 \begin{equation}\label{eq:w_i^*-simple}w_i^* = \begin{cases}
\max_{1\leq j \leq m} v_j^{-2}A_{ji}^2,& p=1, \\
\left(\sum_{j=1}^m v_j^{-q} |A_{ji}|^q\right)^{2/q}, & 1< p < 2, \\
\sum_{j=1}^m v_j^{-2} A_{ji}^2, & p=2 .\end{cases}\end{equation}
These weights can be computed in $O(\nnz(A))$ time. The general formula for $w^*$ for arbitrary blocks and matrices $B_i$ is given in \eqref{eq:w_i}.

Moreover, we show  (Theorems~\ref{thm:beta-dso} and \ref{thm:beta-eso-taunice}) that $(\ff_\mu,\hat{S})\sim \ESO(\beta,w^*)$, where $\beta = \tfrac{\beta'}{\sigma \mu}$ and
\[\beta' = \begin{cases}\min\{\omega,\tau\}, & \text{ if $\hat{S}$ is $\tau$-uniform},\\
1+ \tfrac{(\omega-1)(\tau-1)}{\max\{1,n-1\}}, & \text{ if $\hat{S}$ is $\tau$-nice and $p=2$},\end{cases}\]
and $\omega$ is the degree of Nesterov separability. The formula for $\beta'$ in the case of a $\tau$-nice sampling $\hat{S}$ and $p=1$ is more involved and is given in Theorem~\ref{thm:beta-eso-taunice}. This value is always larger than $\beta'$ in the $p=2$ case (recall that small $\beta'$ is better), and increases with $m$. However, they are often very close in practice (see Figure~\ref{fig:compeso}).

Surprisingly, the formulas for $\beta'$ in the two cases summarized above are identical to those obtained in \cite{RT:PCDM} for smooth partially separable functions (recall the discussion in Section~\ref{sec:generic}), although the classes of functions considered are \emph{different}.
The investigation of this phenomenon is an open question.

We also give formulas for $\beta$ for arbitrary $w$, but these involve the computation of a complicated matrix norm (Theorem~\ref{thm:beta-eso-general}). The above formulas for $\beta$ are \emph{good} (in terms of the parallelization speedup they lead to), \emph{easily computable} and \emph{interpretable} bounds on this norm for $w=w^*$.

\item \textbf{Complexity.} Our complexity results are spelled out in detail in Theorems~\ref{thm:complexity_strong_convexity_smoothed} and \ref{thm:complexity_strong_convexity_nonsmooth}, and are summarized in the table below.

\begin{center}
\begin{tabular}{|c|c|c|}
  \hline
                           & strong convexity & convexity  \\
  \hline
  & & \phantom{aaa}\\
  Problem \ref{eq:P} [Thm~\ref{thm:complexity_strong_convexity_smoothed}]       & $\tfrac{n}{\tau}\times \tfrac{ \tfrac{\beta'}{\mu \sigma} + \sigma_\Psi}{\sigma_{f_\mu} + \sigma_\Psi}$ & $\tfrac{n \beta'}{\tau} \times \tfrac{2 Diam^2}{\mu \sigma \epsilon}$\\
  & & \phantom{aaa} \\
  Problem \ref{eq:Psmooth} [Thm~\ref{thm:complexity_strong_convexity_nonsmooth}] & $\frac{n}{\tau}\times \frac{ \tfrac{2\beta' D}{\epsilon \sigma} + \sigma_\Psi}{\sigma_{f_\mu} + \sigma_\Psi}$ & $\frac{n \beta'}{\tau} \times \frac{8 D Diam^2}{\sigma \epsilon^2}$\\
  & & \phantom{aaa}\\
  \hline
\end{tabular}
\end{center}

The results are complete up to logarithmic factors and say that as long as SPCDM takes at least $k$ iterations, where lower bounds for $k$ are given in the table, then $x_k$ is an $\epsilon$-solution with probability at least $1-\rho$. The confidence level parameter $\rho$ can't be found in the table as it appears in a logarithmic term which we suppressed from the table. For the same reason, it is easy for SPCDM to achieve arbitrarily high confidence. More on the parameters: $n$ is then umber of blocks, $\sigma, \mu$ and $D$ are defined in \S2 of Section~\ref{Sec:8sjs8s}. The remaining parameters will be defined precisely in Section~\ref{sec:complexity}: $\sigma_\phi$ denotes the strong convexity constant of $\phi$ with respect to the norm $\|\cdot\|_{w^*}$ (for $\phi=\Psi$ and $\phi=\ff_\mu$) and $Diam$ is the diameter of the level set of the loss function defined by the value of the loss function at the initial iterate $x_0$.

Observe that as $\tau$ increases, the number of iteration decreases. The actual rate of decrease is controlled by the value of $\beta'$ (as this is the only quantity that may grow with $\tau$). In the convex case, any value of $\beta'$ smaller than $\tau$ leads to parallelization speedup. Indeed, as we discussed in \S3 above, the values of $\beta'$ are much smaller than $\tau$, and decrease to $1$ as $\omega$ approaches $1$. Hence, the more separable the problem is, in terms of the degree of partial separability $\omega$, the better. In the strongly convex case, the situation is even better.

\item \textbf{Cost of a single iteration.} The arithmetic cost of a single iteration of SPCDM is $c = c_1+c_2+c_3$, where $c_1$ is the cost of
computing the gradients $(\nabla f(x_k))^{(i)}$ for $i \in S_k$, $c_2$ is the cost of computing the updates $h_k^{(i)}$ for $i \in S_k$, and $c_3$ is the cost of applying these updates. For simplicity, assume that all blocks are of size 1 and that we update $\tau$ blocks at each iteration. Clearly, $c_3 = \tau$. Since often $h_k^{(i)}$ can be computed in closed form\footnote{This is the case in many cases, including  i) $\Psi_i(t) = \lambda_i|t|$, ii) $\Psi_i(t) = \lambda_i t^2$, and iii) $\Psi_i(t) = 0$ for $t \in [a_i,b_i]$ and $+\infty$ outside this interval (and the multivariate/block generalizations of these functions). For complicated functions $\Psi_i(t)$ one may need to do one-dimensional optimization, which will cost $O(1)$ for each $i$, provided that we are happy with an inexact solution. An analysis of PCDM in the $\tau=1$ case in such an  inexact setting can be found in Tappenden et al \cite{TRG:InexactCDM}, and can be extended to the parallel setting.} and takes $O(1)$ operations, we have $c_2= O(\tau)$. The value of $c_1$ is more difficult to predict in general since  by Proposition~\ref{prop:Nesterov}, we have \[\nabla \ff_\mu(x_{k})=A^T z_k,\] where $z_{k} = \arg \max_{z\in Q} \left\{\ve{Ax_{k}}{z} - g(z) -\mu d(z)\right\}$, and hence $c_1$ depends on the relationship between $A,Q,g$ and $d$. It is often the case though that $z_{k+1}$ is obtained from $z_k$ by changing at most $\delta$ coordinates, with  $\delta$ being small. In such a case it is efficient to maintain the vectors $\{z_{k}\}$ (update at each iteration will cost $\delta$)  and at iteration $k$ to compute $(\nabla\ff_\mu(x_k) ) ^{(i)} = (A^T z_k)^{(i)} = \ve{a_i}{z_k}$ for $i \in S_k$, where $a_i$ is the $i$-th column of $A$, whence  $c_1 = \delta + 2\sum_{i \in S_k}\|a_i\|_0$. Since $\Prob(i \in S_k) = \tau/n$, we have \[\Exp[c_1] = \delta +  \tfrac{2\tau}{n}\sum_{i=1}^n  \|a_i\|_0  = \delta +  \tfrac{2\tau}{n} \nnz(A).\]
In summary, the expected overall arithmetic cost of a single iteration of SPCDM, under the assumptions made above, is $\Exp[c] = O(\tfrac{\tau }{n}\nnz(A) + \delta)$.

\item \textbf{Parallel randomized AdaBoost.} We observe that the \emph{logarithm} of the exponential loss function, which is very popular in machine learning\footnote{Schapire and Freund have written a book~\cite{Boosting-book} entirely dedicated to boosting and \emph{boosting methods}, which are serial/sequential greedy coordinate descent methods, independently discovered in the machine learning community. The original boosting method, AdaBoost, minimizes the exponential loss, and it the most famous boosting algorithm.}, is of the form
    \[ \ff_\mu(x)=\log \Big(\tfrac{1}{m}\sum_{j=1}^m \exp(b_j(Ax)_j)\Big). \]
    for $\mu=1$ and $\ff(x) = \max_j b_j(Ax)_j$. SPCDM in this case can be interpreted as a parallel randomized boosting method. More    details are given in Section~\ref{sec:adaboost}, and in a follow up\footnote{The results presented in this paper were obtained the Fall of 2012 and Spring of 2013, the follow up work of Fercoq~\cite{Fer-ParallelAdaboost} was prepared in the Summer of 2013.} paper of Fercoq~\cite{Fer-ParallelAdaboost}. Our complexity results improve on those in the machine learning literature. Moreover, our framework makes  possible the use of regularizers. Note that Nesterov separability in the context of machine learning requires all examples to depend on at most $\omega$ features, which is often the case.

\item \textbf{Big data friendliness.} Our method is suitable for solving \emph{big data}  nonsmooth \eqref{eq:P}  and smooth \eqref{eq:Psmooth} convex composite Nesterov separable problems in cases when $\omega$ is relatively small compared to $n$. The reasons for this are: i) the parameters of our method ($\beta$ and $w=w^*$) can be obtained easily, ii) the cost of a single iteration decreases for smaller $\omega$, iii) the method is equipped with provable parallelization speedup bounds which get better as $\omega$ decreases, iv) many real-life big-data problems are sparse and can be modeled in our framework as problems with small $\omega$, v) we demonstrate through numerical experiments involving preliminary medium-scale experiments involving millions of variables that our methods are scalable and that our theoretical parallelization speedup predictions hold.

\item \textbf{Subspace Lipschitz constants.} We derive simple formulas for Lipschitz constants of the gradient of $\ff_\mu$ associated with subspaces spanned by an arbitrary subset $S$ of blocks (Section~\ref{sec:DSO}). As a special case, we show that the gradient of a Nesterov separable function is Lipschitz with respect to the norm separable $\|\cdot\|_{w^*}$ with constant equal to $\tfrac{\omega}{\sigma \mu}$, where $\omega$ is degree of Nesterov separability. Besides being useful in our analysis, these results are also of independent interest in the design of gradient-based algorithms in big dimensions.

\end{enumerate}

\section{Fast Computation of Subspace Lipschitz Constants} \label{sec:DSO}

Let us start by introducing the key concept of this section.

\begin{defn} Let $\phi: \R^N \to \R$ be a smooth function and let $\emptyset \neq S \subseteq \{1,2,\dots,n\}$. Then we say that $L_S(\nabla \phi)$ is a \emph{Lipschitz constant of $\nabla \phi$ associated with $S$}, with respect to norm $\|\cdot\|$, if
\begin{equation}\label{eq:DSO-def}\phi(x+h_{[S]}) \leq \phi(x) + \langle \nabla \phi(x), h_{[S]} \rangle + \frac{L_S(\nabla \phi)}{2} \|h_{[S]}\|^2, \qquad x,h \in \R^N.\end{equation}
We will alternatively say that $L_S(\nabla \phi)$ is a subspace Lipschitz constant of $\nabla \phi$ \emph{corresponding to the subspace spanned by blocks $i$ for $i \in S$}, that is, $\{\sum_{i \in S} U_i x^{(i)}\st x^{(i)} \in \R^{N_i}\}$, or simply a \emph{subspace Lipschitz constant}.
\end{defn}

Observe the above inequality can can be equivalently written as
\[\phi(x+h) \leq \phi(x) + \langle \nabla \phi(x), h \rangle + \frac{L_{\J(h)}(\nabla \phi)}{2} \|h\|^2, \qquad x,h \in \R^N.\]

In this section we will be concerned with obtaining easily computable formulas for subspace Lipschitz constants for $\phi = \ff_\mu$ with respect to the separable norm $\|\cdot\|_w$. Inequalities of this type were first introduced  in \cite[Section~4]{RT:PCDM} (therein called Deterministic Separable Overapproximation, or DSO). The basic idea is that in a parallel coordinate descent method in which $\tau$ blocks are updated at each iteration,  subspace Lipschitz constants for sets $S$ of cardinality $\tau$ are more relevant (and possibly much smaller = better) than the standard Lipschitz constant of the gradient, which corresponds to the special case $S=\{1,2,\dots,n\}$ in the above definition. This generalizes the concept of block/coordinate Lipschitz constants introduced by Nesterov~\cite{Nesterov:2010RCDM} (in which case $|S|=1$) to spaces spanned by multiple blocks.

We first derive a \emph{generic} bound on subspace Lipschitz constants (Section~\ref{sec:genericDSO}), one that holds for any choice of $w$ and $v$. Subsequently  we  show (Section~\ref{eq:DSO-w^*}) that for a particular data-dependent choice of the parameters $w_1,\dots,w_n>0$ defining the norm in $\R^N$, the generic bound can be written in a very simple form from which it is clear that i) $L_S \leq L_{S'}$ whenever $S\subset S'$ and ii) that $L_S$ decreases as the degree of Nesterov separability $\omega$ decreases. Moreover, it is important that the data-dependent weights $w^*$ and the factor are easily computable, as these parameters are needed to run the algorithm.

\subsection{Primal update spaces}\label{sec:PUS}

As a first step we need to construct a collection of normed spaces associated with the subsets of $\{1,2,\dots,n\}$. These will be needed in the technical proofs and also in the formulation of our results.

\begin{itemize}
\item \textbf{Spaces.} For $\emptyset \neq S \subseteq  \{1,2,\dots,n\}$ we define $\R^S \eqdef \bigotimes_{i \in S} \R^{N_i}$ and for $h\in \R^N$ we write
$h^{(S)}$ for the vector in $\R^S$ obtained from $h$ by deleting all coordinates belonging to blocks $i \notin S$ (and otherwise keeping the order of the coordinates).\footnote{Note that $h^{(S)}$ is different from $h_{[S]}=\sum_{i \in S} U_i h^{(i)}$, which is a vector in $\R^N$, although both $h^{(S)}$ and $h_{[S]}$ are composed of  blocks $h^{(i)}$ for $i \in S$.}

\item \textbf{Matrices.} Likewise, let $A^{(S)} : \R^{S} \to \R^m$ be the matrix obtained from $A\in \R^{m\times N}$ by deleting all columns corresponding to blocks $i \notin S$, and note that
\begin{equation}\label{eq:A^Sh^S} A^{(S)}h^{(S)}  = Ah_{[S]}.\end{equation}

\item \textbf{Norms.} We fix positive scalars $w_1,w_2,\dots, w_n$ and on $\R^S$ define a pair of conjugate norms as follows
\begin{equation}\label{eq:norm_w}\|h^{(S)}\|_w \eqdef \left(\sum_{i\in S} w_i \ve{B_i h^{(i)}}{h^{(i)}}\right)^{1/2}, \qquad
\|h^{(S)}\|_w^* \eqdef \left(\sum_{i\in S} w_i^{-1} \ve{B_{i}^{-1} h^{(i)}}{h^{(i)}}\right)^{1/2}.
\end{equation}
The standard Euclidean norm of a vector $h^{(S)}\in \R^S$ is given by \begin{equation}\label{eq:E-norm0}\|h^{(S)}\|_E^2 = \sum_{i \in S}\|h^{(i)}\|_E^2 = \sum_{i \in S} \ve{h^{(i)}}{h^{(i)}} .\end{equation}
\end{itemize}

{\footnotesize \emph{Remark:}
Note that, in particular, for $S=\{i\}$ we get $h^{(S)} = h^{(i)}\in \R^{N_i}$ and $\R^S \equiv \R^{N_i}$ (primal block space); and for $S=[n]$ we get $h^{(S)}=h\in \R^N$ and $\R^S \equiv \R^N$ (primal basic space). Moreover, for all $\emptyset \neq S\subseteq [n]$ and $h\in\R^N$,
\begin{equation}\label{eq:js8s4s3s}\|h^{(S)}\|_w  = \|h_{[S]}\|_w,\end{equation}
where the first norm is in $\R^S$ and the second in $\R^N$.
}

\subsection{General bound} \label{sec:genericDSO}

Our first  result in this section, Theorem~\ref{thm:DSO1}, is a refinement of inequality \eqref{eq:DSO-Nesterov-corol} for a sparse update vector $h$.
The only change consists in the term $\|A\|_{w,v}^2$  being replaced by $\|A^{(S)}\|_{w,v}^2$, where $S=\J(h)$ and  $A^{(S)}$ is the matrix, defined in Section~\ref{sec:PUS}, mapping vectors in the primal update space $\E_1 \equiv \R^S$ to vectors in the dual basic space $\E_2\equiv \R^m$. The primal and dual norms  are given by $\|\cdot\|_1 \equiv \|\cdot\|_w$ and  $\|\cdot\|_2 \equiv \|\cdot\|_v$, respectively.   This is indeed a refinement, since for any $\emptyset \neq S \subseteq [n]$,

\begin{eqnarray*}
\|A\|_{w,v} &\overset{\eqref{eq:12norms}}{=}& \max_{\substack{\|h\|_w = 1\\h\in \R^N}}\|Ah\|_v^* \\
&\geq&  \max_{\substack{\|h\|_w = 1\\ h^{(i)} = 0,\; i\in S\\h\in \R^N} } \|Ah\|_v^* \overset{\eqref{h_[S]}}{=} \max_{\substack{\|h_{[S]}\|_w = 1\\h\in \R^N}}\|Ah_{[S]}\|_v^* \overset{\eqref{eq:A^Sh^S}+\eqref{eq:js8s4s3s}}{=} \max_{\substack{\|h^{(S)}\|_w = 1\\h\in \R^N}} \|A^{(S)}h^{(S)}\|_v^* \overset{\eqref{eq:12norms}}{=} \|A^{(S)}\|_{w,v}.
\end{eqnarray*}

The improvement can be dramatic, and gets better  for smaller sets $S$; this will be apparent later. Note that in the same manner one can show that $\|A^{(S_1)}\|_{w,v}\leq \|A^{(S_2)}\|_{w,v}$ if $\emptyset \neq S_1\subset S_2$.


\begin{thm}[Subspace Lipschitz Constants] \label{thm:DSO1} For any $x\in \R^N$ and nonzero $h\in \R^N$,
\begin{equation}\label{eq:DSO-Lemma}\ff_\mu(x+h) \leq \ff_\mu(x) + \langle \nabla \ff_\mu(x), h \rangle + \frac{\|A^{(\J(h))}\|^2_{w,v}}{2 \mu \sigma} \|h\|_{w}^2.\end{equation}
\end{thm}

\begin{proof} Fix $x \in \R^N$, $\emptyset \neq S \subseteq [n]$ and define  $\bar{\ff}: \R^S \to \R$ by
\begin{eqnarray}\notag \bar{\ff}(h^{(S)}) \eqdef \ff_\mu(x+h_{[S]}) &=& \max_{u\in \Q}\left\{\ve{A(x+h_{[S]})}{u} - \hphi(u) -\mu d(u)\right\}\\
\label{eq:jasiu9383}  & \overset{\eqref{eq:A^Sh^S}}{=}& \max_{u\in \Q}\left\{\ve{A^{(S)} h^{(S)}}{u} - \bar{\hphi}(u) -\mu d(u)\right\},
\end{eqnarray}
where $\bar{\hphi}(u)=\hphi(u)-\ve{Ax}{u}$. Applying Proposition~\ref{prop:Nesterov} (with $\E_1=\R^S$, $\E_2=\R^m$, $\bar{A}=A^{(S)}$, $\bar{Q}=Q$, $\|\cdot\|_1=\|\cdot\|_w$ and $\|\cdot\|_2=\|\cdot\|_v$), we conclude that
the gradient of $\bar{\ff}$ is Lipschitz with respect to $\|\cdot\|_w$ on $\R^S$, with Lipschitz constant
$\tfrac{1}{\mu \sigma}\|A^{(S)}\|^2_{w,v}$. Hence, for all $h \in \R^N$,
\begin{equation}\label{eq:js8jd88d}\ff_\mu(x+h_{[S]}) = \bar{\ff}(h^{(S)}) \leq \bar{\ff}(0) + \ve{\nabla \bar{\ff}(0)}{h^{(S)}} + \frac{\|A^{(S)}\|_{w,v}^2}{2\mu \sigma} \|h^{(S)}\|_{w}^2 \enspace.
\end{equation}

Note that $\nabla \bar{\ff}(0) = (A^{(S)})^T u^*$ and $\nabla \ff_\mu(x) = A^T u^*$, where $u^*$ is the maximizer in \eqref{eq:jasiu9383}, whence
\[\ve{\nabla \bar{\ff}(0)}{h^{(S)}} = \ve{(A^{(S)})^Tu^*}{h^{(S)}} = \ve{u^*}{A^{(S)}h^{(S)}} \overset{\eqref{eq:A^Sh^S}}{=} \ve{u^*}{A h_{[S]}} = \ve{A^Tu^*}{ h_{[S]}} = \ve{\nabla \ff_\mu(x)}{h_{[S]}}.\]
Substituting this and the identities $\bar{\ff}(0) = \ff_\mu(x)$ and \eqref{eq:js8s4s3s}  into \eqref{eq:js8jd88d} gives

\[\ff_\mu(x+h_{[S]})  \leq \ff_\mu(x) + \ve{\nabla \ff_\mu(x)}{h_{[S]}} + \frac{\|A^{(S)}\|_{w,v}^2}{2\mu \sigma} \|h_{[S]}\|_{w}^2.
\]


It now remains to observe that in view of \eqref{eq:jd8307d} and \eqref{h_[S]}, for all $h\in \R^N$ we have $h_{[\J(h)]} = h$.
\end{proof}

\subsection{Bounds for data-dependent weights $w$}\label{eq:DSO-w^*}

From now on we will not consider arbitrary weight vector $w$ but one defined by the data matrix $A$ as follows. Let us define  $w^* = (w_1^*,\dots, w_n^*)$ by
\begin{equation}\label{eq:w_i}
w_i^* \eqdef \max \{ (\norm{A_i B_i^{-1/2} t}_v^*)^2 \st t \in \R^{N_i},\; \|t\|_E = 1 \}, \quad i=1,2,\dots,n.
\end{equation}
Notice that as long as the matrices $A_1,\dots, A_n$ are nonzero, we have $w_i^*>0$ for all $i$, and hence the norm $\|\cdot\|_1=\|\cdot\|_{w^*}$ is well defined. 
When all blocks are of size 1 (i.e., $N_i=1$ for all $i$) and $B_i=1$ for all $i$, this reduces to \eqref{eq:w_i^*-simple}. Let us return to the general block setting. Letting $S = \{i\}$ and $\|\cdot\|_1\equiv \|\cdot\|_{w^*}$, we see that $w^*_i$ is defined so that the $\|A^{(S)}\|_{w^*,v}=1$. Indeed,
\begin{eqnarray}
\norm{A^{(S)}}_{w^*,v}^2 & \overset{\eqref{eq:12norms}}{=} & \max_{\|h^{(S)}\|_{w^*} = 1} (\norm{A^{(S)} h^{(S)} }_v^*)^2 \quad \overset{\eqref{eq:A^Sh^S}+\eqref{h_[S]}}{=} \quad \max_{\|h^{(i)}\|_{w^*} = 1} (\norm{A U_i h^{(i)}}_v^*)^2  \notag\\
&\overset{\eqref{eq:norm_block} + \eqref{eq:norm_w} }{=}&  \tfrac{1}{w_i^*} \max_{\|y^{(i)}\|_E = 1}  (\norm{A U_i B_i^{-1/2} y^{(i)}}_v^*)^2 \quad \overset{\eqref{eq:w_i}}{=} \quad  1.\label{eq:9879X}
\end{eqnarray}

In the rest of this section we establish an easily computable upper bound on $\|A^{(\J(h))}\|^2_{w^*,v}$ which will be useful in proving a complexity result for SPCDM used with a $\tau$-uniform or $\tau$-nice sampling. The result is, however, of  independent interest, as we argue at the end of this section.

The following is a technical lemma needed to establish the main result of this section.

\begin{lem}\label{lem:666} 
For any $ \emptyset \neq S \subseteq [n]$ and $w^*$ chosen as in \eqref{eq:w_i}, the following hold:
\begin{eqnarray*}p = 1 \qquad & \Rightarrow & \qquad \max_{\norm{h^{(S)}}_{w^*}=1} \max_{1\leq j\leq m} v_j^{-2} \sum_{i \in S}
(A_{ji}  h^{(i)})^2 \leq 1,\\
1 < p \leq 2 \qquad & \Rightarrow & \qquad \max_{\norm{h^{(S)}}_{w^*}=1} \sum_{j=1}^m  \left( v_j^{-q} \sum_{i \in S}
(A_{ji}  h^{(i)})^2\right)^{q/2} \leq 1.\\
\end{eqnarray*}
\end{lem}

\begin{proof} For any $h^{(i)}$ define the transformed variable $y^{(i)} = (w_i^*)^{1/2} B_i^{1/2} h^{(i)}$ and note that
\[\|h^{(S)}\|_{w^*}^2 \overset{\eqref{eq:norm_w}+
\eqref{eq:norm_block}}{=}  \sum_{i \in S} w_i^* \ve{B_i h^{(i)}}{h^{(i)}}  = \sum_{i \in S} \ve{y^{(i)}}{y^{(i)}} \overset{\eqref{eq:E-norm0}}{=} \|y^{(S)}\|_E^2.\]
We will now prove the result separately for $p=1$, $p=2$ and $1<p<2$.  For $p=1$ we have
\begin{eqnarray}
LHS & \eqdef & \max_{\norm{h^{(S)}}_{w^*}=1} \max_{1\leq j\leq m} v_j^{-2} \sum_{i \in S}
(A_{ji}  h^{(i)})^2 \notag\\
& = & \max_{\|y^{(S)}\|_E=1} \max_{1\leq j \leq m} \left(v_j^{-2}  \sum_{i \in S} (w_i^*)^{-1} (A_{ji} B_i^{-1/2} y^{(i)})^2\right)\notag\\
 &\leq &  \max_{\norm{y^{(S)}}_E=1} \left(\sum_{i \in S} (w_i^*)^{-1} \max_{1\leq j \leq m} \left(v_j^{-2} (A_{ji} B_i^{-1/2} y^{(i)})^2\right)\right) \notag\\
&= & \max_{\norm{y^{(S)}}_E=1} \left(\sum_{i \in S} \|y^{(i)}\|_E^2  \underbrace{(w_i^*)^{-1}  \max_{1\leq j \leq m} \left(v_j^{-2} \left(A_{ji} B_i^{-1/2} \tfrac{y^{(i)}}{\|y^{(i)}\|_E}\right)^2\right)}_{\leq \|A^{(\{i\})}\|_{w^*,v}^2 = 1}\right)\notag\\
&\overset{\eqref{eq:9879X}}{\leq}&  \max_{\norm{y^{(S)}}_E=1} \sum_{i \in S} \|y^{(i)}\|_E^2 \quad \overset{\eqref{eq:E-norm0}}{=} \quad 1 \notag.
\end{eqnarray}

For $p>1$ we may write:
\begin{equation}
 LHS  \eqdef    \max_{\norm{h^{(S)}}_{w^*}=1} \sum_{j=1}^m v_j^{-q} \left(\sum_{i \in S}
(A_{ji}h^{(i)})^2\right)^{q/2} =  \max_{\|y^{(S)}\|_E=1} \sum_{j=1}^m v_j^{-q}  \left(\sum_{i \in S} (w_i^*)^{-1} (A_{ji} B_i^{-1/2} y^{(i)})^2\right)^{q/2}. \label{eq:hjgas78sX}
\end{equation}

In particular, for $p=2$ (i.e., $q=2$) we now have
\begin{eqnarray*}
LHS & \overset{\eqref{eq:hjgas78sX}}{=}&  \max_{\norm{y^{(S)}}_E=1} \sum_{i \in S} (w_i^*)^{-1} \sum_{j=1}^m v_j^{-2} (A_{ji} B_i^{-1/2} y^{(i)})^2 \\
&= & \max_{\norm{y^{(S)}}_E=1} \sum_{i \in S} \|y^{(i)}\|_E^2  \underbrace{(w_i^*)^{-1}  \sum_{j=1}^m v_j^{-2} \left(A_{ji} B_i^{-1/2} \tfrac{y^{(i)}}{\|y^{(i)}\|_E}\right)^2}_{\leq \|A^{(\{i\})}\|_{w^*,v}^2 = 1}\\
&\overset{\eqref{eq:9879X}}{\leq}& \max_{\norm{y^{(S)}}_E=1} \sum_{i \in S} \|y^{(i)}\|_E^2  \quad \overset{\eqref{eq:E-norm0}}{=} \quad 1.
\end{eqnarray*}

For $1<p<2$ we will continue\footnote{The proof works for $p=2$ as well, but the one we have given for  $p=2$ is simpler, so we included it.} from \eqref{eq:hjgas78sX}, first by bounding   $R \eqdef \sum_{i \in S} (w_i^*)^{-1} (A_{ji} B_i^{-1/2} y^{(i)})^2 $ using the H\"{o}lder inequality in the form
\[\sum_{i \in S} a_i b_i \leq \left(\sum_{i\in S} |a_i|^s\right)^{1/s} \left(\sum_{i\in S} |b_i|^{s'}\right)^{1/s'},\]
with $a_i = (w_i^*)^{-1}\left(A_{ji}B_i^{-1}\tfrac{y^{(i)}}{\|y^{(i)}\|_E}\right)^2\|y^{(i)}\|^{2-2/s'}$, $b_i = \|y^{(i)}\|_E^{2/s'}$,
$s = q/2$ and $s'=q/(q-2)$.
\begin{eqnarray}\notag
R^{q/2}
&\leq& \left( \sum_{i \in S} (w_i^*)^{-q/2} \left|A_{ji}B_i^{-1}\tfrac{y^{(i)}}{\|y^{(i)}\|_E}\right|^{q} \|y^{(i)}\|_E^2\right) \times \underbrace{\left(\sum_{i \in S}\|y^{(i)}\|_E^{2} \right)^{(q-2)q/4}}_{\leq 1}\\
\label{eq:sg68sjs545sX}&\leq & \sum_{i \in S} (w_i^*)^{-q/2} \left|A_{ji}B_i^{-1}\tfrac{y^{(i)}}{\|y^{(i)}\|_E}\right|^{q} \|y^{(i)}\|_E^2.
\end{eqnarray}

We now substitute \eqref{eq:sg68sjs545sX} into \eqref{eq:hjgas78sX} and continue as in the $p=2$ case:
\begin{eqnarray*}
LHS & \overset{\eqref{eq:hjgas78sX}+\eqref{eq:sg68sjs545sX}}{\leq} &  \max_{\|y^{(S)}\|_E=1} \left(\sum_{j=1}^m v_j^{-q}  \sum_{i \in S} (w_i^*)^{-q/2} \left|A_{ji}B_i^{-1}\tfrac{y^{(i)}}{\|y^{(i)}\|_E}\right|^{q} \|y^{(i)}\|_E^2\right)^{2/q}\\
 &=& \max_{\|y^{(S)}\|_E=1} \left( \sum_{i \in S} \|y^{(i)}\|_E^2 \underbrace{(w_i^*)^{-q/2} \sum_{j=1}^m v_j^{-q}  \left|A_{ji}B_i^{-1}\tfrac{y^{(i)}}{\|y^{(i)}\|_E}\right|^{q}}_{\leq \left(\|A^{(\{i\})}\|_{w^*,v}^{2}\right)^{1/q} \leq 1} \right)^{2/q}\\
&\overset{\eqref{eq:9879X}}{\leq}&  \max_{\norm{y^{(S)}}_E=1} \left(\sum_{i \in S} \|y^{(i)}\|_E^2 \right)^{2/q} =   \left(\max_{\norm{y^{(S)}}_E=1} \sum_{i \in S} \|y^{(i)}\|_E^2 \right)^{2/q} \quad  \overset{\eqref{eq:E-norm0}}{=} \quad 1.
\end{eqnarray*}

\end{proof}

Using the above lemma we can now give a simple and easily interpretable bound on $\|A^{(S)}\|_{w^*,v}^2$.

\begin{lem} \label{lem:090909} For any $\emptyset \neq S \subseteq [n]$ and $w^*$ chosen as in \eqref{eq:w_i},
\[\|A^{(S)}\|_{w^*,v}^2 \leq \max_{1\leq j \leq m} |\J(A^T e_j) \cap S|.\]
\end{lem}
\begin{proof} It will  be useful to note that
\begin{equation}\label{eq:nsj9292} e_j^T A^{(S)}h^{(S)}  \overset{\eqref{eq:A^Sh^S}}{=} e_j^T A h_{[S]} \overset{\eqref{h_[S]}+\eqref{eq:A_ji}}{=} \sum_{i \in S} A_{ji} h^{(i)}.\end{equation}
We will (twice) make use the following form of the Cauchy-Schwarz inequality: for scalars $a_i$, $i\in Z$, we have $(\sum_{i\in Z} a_i)^2 \leq |Z|\sum_{i \in Z} a_i^2$. For $p=1$, we have
\begin{eqnarray}
 \norm{A^{(S)}}_{w^*,v}^2 &\overset{\eqref{eq:12norms}}{=} &  \max_{\|h^{(S)}\|_{w^*} \leq 1}(\|A^{(S)} h^{(S)}\|_v^*)^2 \quad \overset{\eqref{eq:q-norm}}{=} \quad \max_{\norm{h^{(S)}}_{w^*}=1} \max_{1\leq j \leq m} v_j^{-2} \left({e_j^T A^{(S)}h^{(S)}}\right)^2 \notag\\
& \overset{\eqref{eq:nsj9292}+\eqref{eq:subvector}}{=} &  \max_{\norm{h^{(S)}}_{w^*}=1} \max_{1\leq j \leq m} v_j^{-2} \left({\sum_{i \in \J(A^T e_j) \cap S} A_{ji} h^{(i)}}\right)^2 \notag\\
& \overset{\text{(Cauchy-Schwarz)}}{\leq} &  \max_{\norm{h^{(S)}}_{w^*}=1} \max_{1\leq j \leq m} \left(v_j^{-2} |\J(A^T e_j) \cap S| \sum_{i \in \J(A^T e_j) \cap S} (A_{ji} h^{(i)})^2 \right)\notag\\
& \leq & \max_{1 \leq j \leq m} |\J(A^T e_j) \cap S| \times \max_{\|h^{(S)}\|_{w^*}=1} \max_{1\leq j \leq m} \left(v_j^{-2}  \sum_{i \in S}  (A_{ji} h^{(i)})^2\right)\notag\\
& \overset{(\text{Lemma}~\ref{lem:666})}{\leq} &  \max_{1\leq j\leq m} | \J(A^T e_j) \cap S|.\notag
\end{eqnarray}

For $1<p\leq 2$, we may write
\begin{eqnarray}
 \norm{A^{(S)}}_{w^*,v}^2
 &\overset{\eqref{eq:12norms}}{=} &  \max_{\|h^{(S)}\|_{w^*} \leq 1}(\|A^{(S)} h^{(S)}\|_v^*)^2 \quad \overset{\eqref{eq:q-norm}}{=} \quad \max_{\norm{h^{(S)}}_{w^*}=1} \left( \sum_{j=1}^m v_j^{-q} \left|{e_j^T A^{(S)}h^{(S)}}\right|^q \right)^{1/q} \notag\\
& \overset{\eqref{eq:nsj9292}+\eqref{eq:subvector}}{=} & \max_{\norm{h^{(S)}}_{w^*}=1} \left( \sum_{j=1}^m v_j^{-q} \left( \left|\sum_{i \in \J(A^T e_j) \cap S} A_{ji} h^{(i)}\right|^2\right)^{q/2} \right)^{2/q} \notag\\
& \overset{\text{(Cauchy-Schwarz)}}{\leq} & \max_{\norm{h^{(S)}}_{w^*}=1} \left(\sum_{j=1}^m v_j^{-q} \left(\left|\J(A^T e_j) \cap S \right| \sum_{i \in   \J(A^T e_j) \cap S}
(A_{ji}h^{(i)})^2\right)^{q/2} \right)^{2/q} \notag\\
& \leq & \max_{1 \leq j \leq m} |\J(A^T e_j) \cap S| \times \max_{\|h^{(S)}\|_{w^*}=1} \left(\sum_{j=1}^m v_j^{-q}  \left(\sum_{i \in S}  (A_{ji} h^{(i)})^2\right)^{q/2} \right)^{2/q}\notag \\
& \overset{(\text{Lemma}~\ref{lem:666})}{\leq} & \max_{1 \leq j \leq m} | \J(A^T e_j) \cap S|.\notag
\end{eqnarray}

\end{proof}

We are now ready to state and prove the  main result of this section. It says that  the (interesting but somewhat non-informative)  quantity $\|A^{(\J(h))}\|_{w,v}^2$ appearing in Theorem~\ref{thm:DSO1} can for $w=w^*$ be bounded by a very natural and easily computable quantity capturing the interplay between the sparsity pattern of the rows of $A$ and the sparsity pattern of $h$.

\begin{thm}[Subspace Lipschitz Constants for $w=w^*$] \label{thm:dso}
For $S\subseteq \{1,2,\dots,n\}$ let \begin{equation}\label{eq:theta}L_S \eqdef \max_{1 \leq j \leq m} \abs{\J(A^T e_j) \cap S}.\end{equation}
Then for all $x,h\in \R^N$,
\begin{equation}\label{eq:DSO2}\ff_\mu(x+h) \leq \ff_\mu(x) + \langle \nabla \ff_\mu(x), h \rangle + \frac{L_{\J(h)}}{2 \mu \sigma} \|h\|_{w^*}^2.\end{equation}

\end{thm}

\begin{proof} In view of Theorem~\ref{thm:DSO1}, we only need to show that $\|A^{(\J(h))}\|_{w^*,v}^2 \leq L_{\J(h)}$. This directly follows from Lemma~\ref{lem:090909}.
\end{proof}

Let us now comment on the meaning of this theorem:
\begin{enumerate}

\item Note that $L_{\J(h)}$ depends on $A$ and $h$ through their \emph{sparsity pattern} only. Furthermore, $\mu$ is a user chosen parameter and $\sigma$ depends on $d$ and the choice of the norm $\|\cdot\|_v$, which is independent of the data matrix $A$. Hence, the term $\tfrac{L_{\J(h)}}{\mu \sigma}$ is \emph{independent} of the \emph{values} of $A$ and $h$. Dependence on $A$ is entirely contained in the weight vector $w^*$, as defined in \eqref{eq:w_i}.

\item For each $S$ we have $L_{S} \leq \min\{\max_{1\leq j\leq m}|\J(A^T e_j)|,|S|\} =  \min\{\omega,|S|\} \leq \omega$, where $\omega$ is the degree of Nesterov separability of $\ff$.
\begin{enumerate}
\item By substituting the bound $L_S \leq \omega$ into \eqref{eq:DSO2} we conclude that the gradient of $\ff_\mu$ is Lipschitz with respect to the norm $\|\cdot\|_{w^*}$, with Lipschitz constant equal to $\tfrac{\omega}{\mu \sigma}$.
\item By substituting $U_i h^{(i)}$ in place of $h$ in \eqref{eq:DSO2} (we can also use Theorem~\ref{thm:DSO1}), we observe that the gradient of $\ff_\mu$ is \emph{block Lipschitz} with respect to the norm $\ve{B_i \cdot}{\cdot}^{1/2}$, with Lipschitz constant corresponding to block $i$ equal to $L_i = \frac{w_i^*}{\mu \sigma}$:
     \[\ff_\mu(x+U_ih^{(i)}) \leq \ff_\mu(x) + \ve{\nabla \ff_\mu(x)}{U_i h^{(i)}} + \frac{L_i}{2}\ve{B_i h^{(i)}}{h^{(i)}}, \quad x\in \R^N, \; h^{(i)}\in \R^{N_i}.\]
    \end{enumerate}
\item In some sense it is more natural to use the norm $\|\cdot\|_L^2$ instead of $\|\cdot\|_{w^*}^2$, where $L=(L_1,\dots,L_n)$ are the block Lipschitz constants $L_i = \tfrac{w_i^*}{\mu \sigma}$ of $\nabla \ff_\mu$. If we do this, then although the situation is very different, inequality \eqref{eq:DSO2} is similar to the one given for partially separable smooth functions in \cite[Theorem 7]{RT:PCDM}. Indeed, the weights defining the norm are in both cases equal to the block Lipschitz constants (of $f$ in \cite{RT:PCDM} and of $\ff_\mu$ here). Moreover, the leading term in \cite{RT:PCDM} is structurally comparable to the leading term $L_{\J(h)}$. Indeed, it is equal to $\max_{S} |\J(h) \cap S|$, where the maximum is taken over the block domains $S$ of the constituent functions $f_S(x)$ in the representation of $f$ revealing partial separability: $f(x) = \sum_{S} f_S(x)$.
\end{enumerate}

\section{Expected Separable Overapproximation (ESO)}\label{sec:ESO}

In this section we  compute parameters $\beta$ and $w$ yielding an ESO for the pair $(\phi,\hat{S})$, where $\phi = \ff_\mu$
and $\hat{S}$ is a proper uniform sampling. If inequality \eqref{eq:ESO} holds, we will for simplicity write $(\phi,\hat{S})\sim \ESO(\beta,w)$. Note also that for all $\gamma>0$,
\[(\phi,\hat{S}) \sim \ESO(\beta \gamma,w) \qquad \Longleftrightarrow \qquad (\phi,\hat{S})\sim \ESO(\beta, \gamma w).\]

In Section~\ref{sec:ESO+Lipschitz} we establish a link between ESO for $(\phi,\hat{S})$ and Lipschitz continuity of the gradient of a certain collection of functions. This link will enable us to compute the ESO parameters $\beta, w$ for the smooth approximation of a Nesterov separable function $\ff_\mu$,  needed both for running Algorithm~\ref{alg:pcdm1} and for the complexity analysis. In Section~\ref{sec:DPS} we define certain technical objects that will be needed for further analysis. In Section~\ref{sec:ESO1} we prove a first ESO result, computing $\beta$ for any $w>0$ and \emph{any} proper uniform sampling. The formula for $\beta$ involves the norm of a certain large matrix, and hence is not directly useful as $\beta$ is needed for running the algorithm. Also, this formula does not explicitly exhibit dependence on $\omega$; that is, it is not immediately apparent that $\beta$ will be smaller for smaller $\omega$, as one would expect.
Subsequently, in Section~\ref{sec:ESO2}, we specialize this result to $\tau$-uniform samplings   and then further to the more-specialized $\tau$-nice samplings in Section~\ref{sec:ESO3}. As in the previous section, in these special cases we show that the choice $w=w^*$ leads to very simple closed-form expressions for $\beta$, allowing us to get direct insight into parallelization speedup.

\subsection{ESO and Lipschitz continuity} \label{sec:ESO+Lipschitz}

We will now study the  collection of functions $\hat{\phi}_x : \R^N\to \R$ for $x \in \R^N$ defined by

\begin{equation}\label{eq:hatphi} \hat{\phi}_x(h) \eqdef  \Exp\left[\phi(x+h_{[\hat{S}]})\right].\end{equation}

Let us first establish some basic connections between $\phi$ and $\hat{\phi}_x$.

\begin{lem} \label{eq:lem_eso} Let $\hat{S}$ be any sampling and $\phi : \R^N\to \R$ any function. Then for all $x\in \R^N$
\begin{enumerate}
\item[(i)] if $\phi$ is convex, so is $\hat{\phi}_x$,
\item[(ii)] $\hat{\phi}_x(0) = \phi(x)$,
\item[(iii)] If $\hat{S}$ is proper and uniform, and $\phi : \R^N\to \R$ is continuously differentiable, then  \[\nabla \hat{\phi}_x(0) = \frac{\Exp[|\hat{S}|]}{n} \nabla \phi(x).\]
\end{enumerate}
\end{lem}
\begin{proof} Fix $x\in \R^N$. Notice that $\hat{\phi}_x(h) = \Exp[\phi(x+h_{[\hat{S}]})] = \sum_{S\subseteq [n]} \Prob(\hat{S}=S) \phi(x+U_S h)$,
where $U_S \eqdef \sum_{i\in S} U_i U_i^T$. As $\hat{\phi}_x$ is a convex combination of convex functions, it is convex, establishing (i). Property (ii) is trivial. Finally,
\[\nabla \hat{\phi}_x(0) =  \Exp \left[\nabla \left. \phi(x+h_{[\hat{S}]}) \right|_{h=0}\right] = \Exp\left[U_{\hat{S}}\nabla \phi(x)\right] = \Exp\left[U_{\hat{S}}\right]\nabla \phi(x) = \frac{\Exp[|\hat{S}|]}{n} \nabla \phi(x).\]
The last equality follows from the observation that $U_{\hat{S}}$ is an $N\times N$ binary diagonal matrix with ones in positions $(i,i)$ for $i\in \hat{S}$ only, coupled with \eqref{eq:uniform_samp_basic}.
\end{proof}

We now establish a connection between ESO and a uniform bound in $x$ on the Lipschitz constants of the gradient ``at the origin'' of the functions $\{\hat{\phi}_x,\; x \in \R^N\}$. The result will be used for the computation of the parameters of ESO for Nesterov separable functions.

\begin{thm} \label{eq:ESO-vs-Lip}Let $\hat{S}$ be proper and uniform, and $\phi:\R^N\to \R$ be continuously differentiable. Then the following statements are equivalent:
\begin{itemize}
\item[(i)] $(\phi,\hat{S}) \sim \ESO(\beta,w)$,
\item[(ii)] $\hat{\phi}_x(h) \leq \hat{\phi}_x(0) + \ve{\nabla \hat{\phi}_x(0)}{h}  + \frac{1}{2}\frac{\Exp[|\hat{S}|]\beta}{n}\|h\|_w^2, \qquad x, h\in \R^N.$
\end{itemize}
\end{thm}
\begin{proof} We only need to  substitute \eqref{eq:hatphi} and \text{Lemma}~\ref{eq:lem_eso}(ii-iii) into inequality (ii) and compare the result with \eqref{eq:ESO}.
\end{proof}

\subsection{Dual product space}\label{sec:DPS}

Here we  construct a linear space associated with a fixed block sampling $\hat{S}$, and several derived objects which will depend on the distribution of $\hat{S}$. These objects  will be needed in the proof of Theorem~\ref{thm:beta-eso-general} and in further text.

\begin{itemize}
\item \textbf{Space.} Let $\calP \eqdef \{S \subseteq [n] \st p_S>0\}$, where $p_S \eqdef \Prob(\hat{S}=S)$. The \emph{dual product space} associated with $\hat{S}$ is defined by \[\R^{|\calP|m}\eqdef \bigotimes_{S \in \calP} \R^m.\]

\item \textbf{Norms.} Letting $\uu = \{\uu^S \in \R^m \st S \in \calP\} \in \R^{|\calP|m}$, we now define a pair of conjugate norms in $\R^{|\calP|m}$ associated with $v$ and $\hat{S}$:
\begin{equation}\label{eq:norm2'}
\|\uu\|_{\hat{v}} \eqdef \Big(\sum_{S \in \calP} p_S\|\uu^S\|_v^2\Big)^{1/2}, \qquad \|\uu\|^*_{\hat{v}} \eqdef \max_{\|\uu'\|_{\hat{v}} \leq 1} \ve{\uu'}{\uu} = \Big(\sum_{S \in \calP} p_S^{-1}(\|\uu^S\|_v^*)^2\Big)^{1/2}.
\end{equation}
The notation $\hat{v}$ indicates dependence on both $v$ and $\hat{S}$.

\item \textbf{Matrices.} For each $S \in \cal P$ let \begin{equation}\label{eq:jjd8d88d}\hat{A}^S  \eqdef p_S A \sum_{i \in S}  U_i U_i^T \in \R^{m \times N}.\end{equation}
We now define matrix $\hat{A} \in \R^{|\calP|m \times N}$, obtained by stacking the matrices $\hat{A}^S$, $S \in \calP$, on top of each other (in the same order the vectors $u^S$, $S \in \cal P$ are stacked to form $u\in \R^{|{\cal P}|m}$). The ``hat'' notation indicates that $\hat{A}$ depends on both $A$ and $\hat{S}$. Note that $\hat{A}$ maps vectors from the primal basic space $\E_1 \equiv \R^N$ to vectors in the dual product space $\E_2 \equiv \R^{|\calP|m}$. We use $\|\cdot\|_1 \equiv \|\cdot\|_w$ as the norm in $\E_1$ and $\|\cdot\|_2 \equiv \|\cdot\|_{\hat{v}}$ as the norm in $\E_2$. It will be useful to note that for $h\in \R^N$, and $S \in \cal P$,
\begin{equation}\label{eq:09sddd}
(\hat{A}h)^S = \hat{A}^S h.
\end{equation}
\end{itemize}

\subsection{Generic ESO for proper uniform samplings}\label{sec:ESO1}

Our first ESO result covers all (proper) uniform samplings and is valid for any $w>0$. We give three formulas with three different values of $\beta$. While we could have instead given a single formula with $\beta$ being the minimum of the three values, this will be useful.

\begin{thm}[Generic ESO] \label{thm:beta-eso-general} If $\hat{S}$ is proper and uniform, then
\begin{equation}\label{eq:klkl--09}
 i) \quad (\ff_\mu,\hat{S}) \sim \ESO\left(\frac{n\|\hat{A}\|_{w,\hat{v}}^2}{\mu\sigma\Exp[|\hat{S}|]}, w\right), \qquad ii) \quad (\ff_\mu,\hat{S}) \sim \ESO\left(\frac{n\Exp\left[\|A^{(\hat{S})}\|_{w,v}^2\right]}{\mu\sigma\Exp[|\hat{S}|]}, w\right),
\end{equation}
\begin{equation}\label{eq:klkl--10}
iii) \quad (\ff_\mu,\hat{S}) \sim \ESO\left(\frac{\max_{S \in {\cal P}}\|A^{(S)}\|_{w,v}^2}{\mu\sigma}, w\right).
\end{equation}
\end{thm}

\begin{proof} We will first establish (i). Consider the function
\begin{eqnarray}
\notag \bar{\ff}(h) & \eqdef & \Exp[\ff_\mu(x+h_{[\hat{S}]})]\\
\notag &\overset{\eqref{eq:f_mu}}{=}& \sum_{S \in \calP}  p_S \max_{\uu^S \in \Q} \left\{\ve{A (x+h_{[S]})}{\uu^S}  -\hphi(\uu^S) - \mu d(\uu^S)\right\}\\
\label{eq:mnbvvcstt65}&= & \max_{\{\uu^S \in \Q \st S\in \calP\}} \sum_{S \in \calP}  p_S  \left\{\ve{Ah_{[S]}}{\uu^S} + \langle Ax, \uu^S \rangle -\hphi(\uu^S) - \mu d(\uu^S)\right\}.
\end{eqnarray}
Let $\uu \in \bar{\Q} \eqdef \Q^{|\calP|}\subseteq \R^{|\calP| m}$ and note that
\begin{equation}\label{eq:098sjsdjd009}\sum_{S \in \calP} p_S \ve{A h_{[S]}}{\uu^S} \overset{\eqref{eq:jjd8d88d}+\eqref{h_[S]}}{=}   \sum_{S \in \calP} \ve{\hat{A}^S h}{\uu^S}    \overset{\eqref{eq:09sddd}}{ =} \ve{\hat{A}h}{\uu} .\end{equation}
Furthermore, define  $\bar{\hphi}: \bar{\Q} \to \R$ by $\bar{\hphi}(\uu) \eqdef  \sum_{S \in \calP} p_S(\hphi(\uu^S)-\ve{Ax}{\uu^S})$, and $\bar{d}: \bar{\Q}  \to \R$ by $\bar{d}(\uu) \eqdef \sum_{S \in \calP} p_Sd(\uu^S)$. Plugging all of the above into \eqref{eq:mnbvvcstt65} gives
\begin{equation}\label{eq:99s9s99s}
\bar{\ff}(h) =  \max_{\uu  \in \bar{\Q}} \left\{\ve{\hat{A}h}{\uu}  -\bar{\hphi}(\uu) - \mu \bar{d}(\uu)\right\}.
\end{equation}
It is easy to see that $\bar{d}$ is $\sigma$-strongly convex on $\bar{\Q}$ with respect to the norm  $\|\cdot\|_{\hat{v}}$ defined in \eqref{eq:norm2'}.
Indeed, for any $\uu_1,\uu_2\in \bar{Q}$ and $t\in (0,1)$,
\begin{eqnarray*}
 \bar{d}(t \uu_1+(1-t)\uu_2)
& = & \sum_{S \in \calP} p_S d(t\uu_1^S+(1-t)\uu_2^S) \\
& \leq & \sum_{S \in \calP} p_S
\big( t d(\uu_1^S)+(1-t)d(\uu_2^S) - \frac{\sigma}{2}t(1-t) \norm{\uu_1^S-\uu_2^S}_v^2 \big) \\
&\overset{\eqref{eq:norm2'}}{=} & t\bar{d}(\uu_1) + (1-t)\bar{d}(\uu_2) - \frac{\sigma}{2}t(1-t) \norm{\uu_1-\uu_2}_{\hat{v}}^2.
\end{eqnarray*}
Due to  $\bar{\ff}$ taking on the form \eqref{eq:99s9s99s}, Proposition~\ref{prop:Nesterov}  (used with $\E_1=\R^N$, $\E_2 = \R^{|{\cal P}| m}$,  $\bar{A}=\hat{A}$, $\|\cdot\|_1 =\|\cdot\|_{w}$, $\|\cdot\|_2 = \|\cdot\|_{\hat{v}}$ and $\bar{\sigma}= \sigma$) says that the gradient of $\bar{\ff}$ is Lipschitz with constant $\tfrac{1}{\mu\sigma}\|\hat{A}\|_{w,\hat{v}}^2$. We now only need to applying Theorem~\ref{eq:ESO-vs-Lip},  establishing (i). 

Let us now show (ii)+(iii). Fix $h\in \R^N$, apply Theorem~\ref{thm:DSO1} with $h\leftarrow h_{[\hat{S}]}$ and take expectations. Using identities \eqref{eq:jss8s8s}, we get \[\Exp[\ff_\mu(x+h_{[\hat{S}]})] \leq \ff(x) + \tfrac{\Exp[|\hat{S}|]}{n}\left( \ve{\nabla \ff_\mu(x)}{h} + \tfrac{n\gamma(h)}{2 \mu \sigma \Exp [|\hat{S}|]}\right), \quad \gamma(h) = \Exp\left[\|A^{(\hat{S})}\|_{w,v}^2 \|h_{[\hat{S}]}\|_w^2\right].\] Since $\|h_{[\hat{S}]}\|_w^2\leq \|h\|_w^2$, we have $\gamma(h)\leq \Exp \left[\|A^{(\hat{S})}\|_{w,v}^2\right]\|h\|_w^2$, which establishes (ii). Since $\|A^{(\hat{S})}\|_{w,v}^2 \leq \max_{S \in {\cal P}}\|A^{(S)}\|_{w,v}^2$, using \eqref{eq:jss8s8s} we obtain $\gamma(h)\leq \tfrac{\Exp[|\hat{S}|] \max_{S \in {\cal P}}\|A^{(S)}\|_{w,v}^2}{n}\|h\|_w^2$, establishing (iii).
\qedhere
\end{proof}

We now give an insightful characterization of $\|\hat{A}\|_{w,\hat{v}}$.

\begin{thm} If $\hat{S}$ is proper and uniform, then
\begin{equation}\label{eq:3434334mm}\|\hat{A}\|_{w,\hat{v}}^2 = \max_{h \in \R^{N}, \; \|h\|_w \leq 1} \Exp \left[\left(\|A h_{[\hat{S}]}\|_v^*\right)^2\right].\end{equation}
Moreover,
\[\left(\frac{\Exp[|\hat{S}|]}{n}\right)^2\|A\|_{w,v}^2 \; \leq \; \|\hat{A}\|_{w,\hat{v}}^2 \; \leq \; \min\left\{ \Exp \left[\|A^{(S)}\|_{w,v}^2\right] \;,\; \frac{\Exp [|\hat{S}|]}{n}\|A\|_{w,v}^2 \;,\; \max_{S \in {\cal P}}\|A^{(S)}\|^2_{w,v} \right\}.\]
\end{thm}
\begin{proof}
Identity \eqref{eq:3434334mm} follows from
\begin{eqnarray}
\|\hat{A}\|_{w,\hat{v}}
&\overset{\eqref{eq:12norms}}{=} & \max\{\ve{\hat{A}h}{u} \st \|h\|_w \leq 1, \quad \|u\|_{\hat{v}}\leq 1\}\notag \\
&\overset{\eqref{eq:098sjsdjd009}+\eqref{eq:norm2'}}{=}& \max\left\{\sum_{S \in {\cal P}} p_S \ve{Ah_{[S]}}{u^S} \st \|h\|_w \leq 1, \quad \sum_{S \in {\cal P}} p_S\|u^S\|_{v}^2\leq 1 \right\}. \label{eq:0909ooop}\\
&=& \max_{\|h\|_w \leq 1} \max_{u} \left\{ \sum_{S \in {\cal P}}  p_S \|u^S\|_v \ve{Ah_{[S]}}{\tfrac{u^S}{\|u^S\|_v}} \st \sum_{S \in {\cal P}} p_S\|u^S\|_{v}^2\leq 1 \right\}\notag \\
&=& \max_{\|h\|_w \leq 1} \max_{\beta} \left\{ \sum_{S \in {\cal P}}  p_S \beta_S \|Ah_{[S]}\|_{v}^* \st \sum_{S \in {\cal P}} p_S \beta_S^2 \leq 1, \; \beta_S \geq 0\right\}\notag \\
&=& \max_{\|h\|_w \leq 1} \left( \sum_{S \in {\cal P}}  p_S \left(\|Ah_{[S]}\|_{v}^*\right)^2 \right)^{1/2} \quad = \quad \max_{\|h\|_w\leq 1} \left(\Exp \left[\left(\|A h_{[\hat{S}]}\|_{v}^*\right)^2\right]\right)^{1/2}.\notag
\end{eqnarray}
As a consequence, we now have
\begin{eqnarray*}\|\hat{A}\|_{w,\hat{v}}^2 &\overset{\eqref{eq:3434334mm}}{\leq}&  \Exp \left[ \max_{h \in \R^{N}, \; \|h\|_w \leq 1}\left(\|A h_{[\hat{S}]}\|_v^*\right)^2\right] \\
&\overset{\eqref{eq:A^Sh^S}}{=} & \Exp \left[ \max_{h \in \R^{N}, \; \|h\|_w \leq 1}\left(\|A^{(\hat{S})} h^{(\hat{S})}\|_v^*\right)^2\right] \quad \overset{\eqref{eq:12norms}}{ =} \quad \Exp \left[\|A^{(\hat{S})}\|_{w,v}^2\right] \quad \leq \quad \max_{S \in {\cal P}} \|A^{(S)}\|_{w,v}^2,\end{eqnarray*}
and
\begin{eqnarray*}
\|\hat{A}\|_{w,\hat{v}}^2 &\overset{\eqref{eq:3434334mm}}{\leq} & \max_{h \in \R^N,\; \|h\|_w \leq 1}\Exp\left[\|A\|_{w,v}^2 \|h_{[\hat{S}]}\|_w^2\right] 
\quad \overset{\eqref{eq:jss8s8s}}{=}\quad  \frac{\Exp[|\hat{S}|]}{n}\|A\|_{w,v}^2.
\end{eqnarray*}
Finally, restricting the vectors $\hat{u}^S$, $S \in {\cal P}$, to be equal (to $z$), we obtain the estimate
\begin{eqnarray*}
\|\hat{A}\|_{w,\hat{v}}
&\overset{\eqref{eq:0909ooop}}{\geq } &  \max\{\Exp[\ve{A h_{[\hat{S}]}}{z}] \st \|h\|_w \leq 1, \quad  \|z\|_{v} \leq 1\}\\
&\overset{\eqref{eq:jss8s8s}}{=} & \max\{ \tfrac{\Exp[|\hat{S}|]}{n} \ve{A h}{z} \st \|h\|_w \leq 1, \quad  \|z\|_{v} \leq 1\} \quad \overset{\eqref{eq:12norms}}{=} \quad \frac{\Exp[|\hat{S}|]}{n}\|A\|_{w,v},
\end{eqnarray*}
giving the lower bound.
\end{proof}

Observe that as a consequence of this result,  ESO (i) in Theorem~\ref{thm:beta-eso-general} is always preferable to ESO (ii). In the following section we will utilize ESO (i) for $\tau$-nice samplings and ESO (iii) for the more  general $\tau$-uniform samplings. In particular, we give easily computable upper bounds on $\beta$ in the special case when $w=w^*$.

\subsection{ESO for data-dependent weights $w$}

Let us first establish ESO for $\tau$-uniform samplings and $w=w^*$. 

\begin{thm}[ESO for $\tau$-uniform sampling] \label{thm:beta-dso}
If $\ff$ is Nesterov separable  of degree $\omega$, $\hat{S}$ is a $\tau$-uniform sampling and $w^*$ is chosen as in \eqref{eq:w_i}, then
\[
(\ff_\mu,\hat{S}) \sim \ESO\left(\beta,w^*\right),
\]
where $\beta = \tfrac{\beta'_1}{\mu \sigma}$ and $\beta'_1 \eqdef \min\{\omega, \tau\}$.
\end{thm}
\begin{proof}
This follows from  ESO (iii) in Theorem~\ref{thm:beta-eso-general} in  by using the bound $\|A^{(S)}\|_{w,v}^2 \leq \max_{j}|\J(A^Te_j) \cap S| \leq \min\{\omega,\tau\}$, $S \in {\cal P}$, which follows from Lemma~\ref{lem:090909} and the fact that $|\J(A^T e_j)|\leq \omega$ for all $j$ and $|S|=\tau$ for all $S \in {\cal P}$.
\end{proof}


Before we establish an ESO result for $\tau$-nice samplings, the main result of this section, we need a technical lemma with a number of useful relations. Identities \eqref{eq:0djd7diddd} and \eqref{eq:jsjs898sn} and estimate \eqref{eq:98ss54sgs} are new, the other two identities are from \cite[Section 3]{RT:PCDM}.  For $S\subseteq [n] = \{1,2,\dots,n\}$ define
\begin{equation}\label{eq:chi}
\chi_{(i \in S)} = \begin{cases} 1 & \quad \text{if } i \in S,\\
0 & \quad \text{otherwise.}
\end{cases}
\end{equation}

\begin{lem}\label{lem:14} Let $\hat{S}$ be any sampling, $J_1,J_2$ be nonempty subsets of $[n]$ and $\{\theta_{ij} : \; i\in [n], \; j \in [n]\}$ be any real constants.  Then
\begin{eqnarray}
\label{eq:0djd7diddd}\Exp\left[\sum_{i \in J_1 \cap \hat{S}} \sum_{j \in J_2 \cap \hat{S}} \theta_{ij}\right] &=& \sum_{i \in J_1}\sum_{j \in J_2} \Prob(\{i,j\}\subseteq \hat{S})\theta_{ij}.
\end{eqnarray}
If $\hat{S}$ is $\tau$-nice, then for any $\emptyset \neq J \subseteq [n]$,
 $\theta \in \R^n$ and $k \in [n]$, the following identities hold
\begin{eqnarray}
\label{eq:condexpJinterS}
\Exp\left[ \sum_{i \in J \cap \hat{S}} \theta_i \; | \; \abs{J \cap \hat{S}} = k\right] &=& \frac{k}{\abs{J}}\sum_{i \in J} \theta_i,\\
\label{eq:JinterS2}
 \Exp\left[ \abs{J \cap \hat{S}}^2 \right] &=& \frac{\abs{J} \tau}{n} \Big( 1+\frac{(\abs{J}-1)(\tau-1)}{\max(1, n-1)}\Big) ,\\
\label{eq:jsjs898sn}
\max_{1\leq i \leq n}\Exp[|J \cap \hat{S}| \times \chi_{(i \in \hat{S})}] &=& \frac{\tau}{n} \left(1 + \frac{(|J|-1)(\tau-1)}{\max(1,n-1)}\right).
\end{eqnarray}
Moreover, if $J_1,\dots,J_m$ are subsets of $[n]$ of identical cardinality ($|J_j| =\omega$ for all $j$), then
\begin{equation}\label{eq:98ss54sgs}
\max_{1\leq i \leq n}\Exp[\max_{1\leq j\leq m}|J_j \cap \hat{S}| \times \chi_{(i \in \hat{S})}] \quad \leq \quad \frac{\tau}{n} \sum_{k=1}^{k_{\text{max}}} \min \left\{1 \; , \; \frac{mn}{\tau} \sum_{l=\max\{k,k_{\text{min}}\}}^{k_{\text{max}}} c_l \pi_l \right\},
\end{equation}
where $k_{\text{min}} = \max\{1,\tau-(n-\omega)\}$, $k_{\text{max}} = \min\{\tau,\omega\}$, $c_l = \max\left\{\frac{l}{\omega},\frac{\tau-l}{n-\omega}\right\}\leq 1$ if $\omega<n$ and $c_l = \tfrac{l}{\omega}\leq 1$ otherwise, and
\[\pi_l \eqdef \Prob(|J_j \cap \hat{S} | = l) = \frac{\binom{\omega}{k} \binom{n-\omega}{\tau-k}}{\binom{n}{\tau}}, \qquad k_{\text{min}} \leq l \leq k_{\text{max}}.\]

\end{lem}

\begin{proof}
The first statement is a straightforward generalization of (26) in \cite{RT:PCDM}. Identities \eqref{eq:condexpJinterS} and \eqref{eq:JinterS2} were established\footnote{In fact, the proof of the former is essentially identical to the proof of \eqref{eq:0djd7diddd}, and \eqref{eq:JinterS2} follows from \eqref{eq:0djd7diddd} by choosing $J_1=J_2=J$ and $\theta_{ij} = 1$.} in \cite{RT:PCDM}. Let us prove \eqref{eq:jsjs898sn}. The statement is trivial for $n=1$, assume therefore that $n\geq 2$. Notice that
\begin{equation}\label{eq:9dj9hsdnjcis}\Exp[|J \cap \hat{S}| \times \chi_{(k \in \hat{S})}] =  \Exp\left[\sum_{i \in J \cap \hat{S}} \sum_{j \in \{k\}\cap \hat{S}} 1\right] \overset{\eqref{eq:0djd7diddd}}{=} \sum_{i \in J} \Prob(\{i,k\} \subseteq \hat{S}).\end{equation}
Using \eqref{eq:uniform_samp_basic}, and the simple fact that $\Prob(\{i,k\}\subseteq \hat{S}) = \tfrac{\tau(\tau-1)}{n(n-1)}$ whenever $i\neq k$, we get
\begin{equation}\label{eq:2casesiu98} \sum_{i \in J} \Prob(\{i,k\} \subseteq \hat{S}) =\begin{cases}\sum_{i \in J} \tfrac{\tau(\tau-1)}{n\max(1,n-1)} = \tfrac{|J|\tau(\tau-1)}{n(n-1)}, & \quad  \text{if } k\notin J,\\
\tfrac{\tau}{n} + \sum_{i \in J/\{k\}} \tfrac{\tau(\tau-1)}{n(n-1)}  = \tfrac{\tau}{n}\left(1+\tfrac{(|J|-1)(\tau-1)}{(n-1)}\right), & \quad \text{if } k \in J.\end{cases}\end{equation}
Notice that the expression in the $k \notin J$ case is smaller than expression in the $k \in J$ case. If we now combine \eqref{eq:9dj9hsdnjcis} and \eqref{eq:2casesiu98} and take maximum in $k$, \eqref{eq:jsjs898sn} is proved. Let us now establish \eqref{eq:98ss54sgs}. Fix $i$ and let $\eta_j \eqdef |J_j \cap \hat{S}|$. We can now estimate
\begin{eqnarray}
 \Exp[  \max_{1\leq j \leq m} \eta_j \times \chi_{(i \in\hat{S})} ] &=& \sum_{k=k_{\text{min}}}^{k_{\text{max}}} k \Prob\left( \max_{1\leq j \leq m} \eta_j \times \chi_{(i\in\hat{S})} = k \right) \notag\\
 & = &\sum_{k=1}^{k_{\text{max}}} \Prob\left( \max_{1\leq j \leq m} \eta_j \times \chi_{(i\in\hat{S})} \geq k \right) \notag\\
&=& \sum_{k=1}^{k_{\text{max}}} \Prob\left( \bigcup_{j=1}^m \left\{ \eta_j \geq k  \;\; \& \;\; i\in\hat{S} \right\} \right)\notag\\
&\leq & \sum_{k=1}^{k_{\text{max}}} \min \left\{ \Prob (i \in \hat{S}), \sum_{j=1}^m \Prob\left( \eta_j \geq k \;\; \& \;\; i \in \hat{S} \right) \right\}  \notag\\
&\overset{\eqref{eq:uniform_samp_basic}}{=}& \sum_{k=1}^{k_{\text{max}}} \min \left\{ \frac{\tau}{n}, \sum_{j=1}^m \sum_{l=\max\{k,k_{\text{min}}\}}^{k_{\text{max}}} \Prob\left( \eta_j = l \;\; \& \; \; i\in \hat{S} \right)\right\}.\label{eq:09sns6sggdh}
\end{eqnarray}
In the last step we have used the fact that $\Prob(\eta_j=l) = 0$ for $l<k_{\text{min}}$ to restrict the scope of $l$. Let us now also fix $j$ and estimate  $\Prob( \eta_j = l \;\; \& \; \; i \in \hat{S})$. Consider two cases:
\begin{enumerate}
\item[(i)] If $i \in J_j$, then among the $\binom{n}{\tau}$ equiprobable possible outcomes of the $\tau$-nice sampling $\hat{S}$,
the ones for which $\abs{J_j \cap \hat{S}}=l$ and $i \in \hat{S}$ are those that select block $i$ and
$l-1$ other blocks from $J_j$ ($\binom{\omega-1}{l-1}$ possible choices) and $\tau-l$ blocks from outside $J_j$
($\binom{n-\omega}{\tau-l}$ possible choices). Hence,
\begin{equation}\label{eq:jds6kd0}
\Prob\left( \eta_j = l \;\; \& \;\; i \in \hat{S} \right) = \frac{\binom{\omega-1}{l-1} \binom{n-\omega}{\tau-l}}{\binom{n}{\tau}} = \frac{l}{\omega}\pi_l. \end{equation}

\item[(ii)]
If $i \not \in J_j$ (notice that this can not happen if $\omega=n$), then among the $\binom{n}{\tau}$ equiprobable possible outcomes of the $\tau$-nice sampling $\hat{S}$, the ones for which $\abs{\hat{S} \cap J_j}=l$ and $i \in \hat{S}$ are those that select block $i$ and
$\tau - l -1$ other blocks from outside $J_j$ ($\binom{n-\omega-1}{\tau-l-1}$ possible choices) and $l$ blocks from $J_j$
($\binom{\omega}{l}$ possible choices). Hence,
\begin{equation}\label{eq:kdndhdd79d}
\Prob\left( \eta_j = l \;\; \& \;\; i \in \hat{S} \right) = \frac{\binom{\omega}{l} \binom{n-\omega-1}{\tau-l-1}}{\binom{n}{\tau}} = \frac{\tau- l}{n-\omega}\pi_l.
\end{equation}
\end{enumerate}
It only remains to plug the maximum of \eqref{eq:jds6kd0} and \eqref{eq:kdndhdd79d} into \eqref{eq:09sns6sggdh}.
\end{proof}

We are now ready to present the main result of this section.

\begin{thm}[ESO for $\tau$-nice sampling] \label{thm:beta-eso-taunice}
Let $\ff$ be Nesterov separable of degree $\omega$, $\hat{S}$ be $\tau$-nice, and $w^*$ be chosen as in \eqref{eq:w_i}.
Then
\[
 (\ff_\mu,\hat{S}) \sim \ESO(\beta,w^*),
\]
where $\beta = \frac{\beta'}{\mu\sigma }$ and
\begin{equation}\label{beta_1}\beta'  = \beta'_2 \eqdef  1+\frac{(\omega-1)(\tau-1)}{\max(1, n-1)}\end{equation}
if the dual norm $\|\cdot\|_v$ is defined with $p=2$, and
\begin{equation}\label{eq:p=1beta}\beta' = \beta'_3 \eqdef \sum_{k=1}^{k_{\text{max}}} \!\! \min \left\{ 1, \frac{mn}{\tau} \sum_{l=\max\{k,k_{\text{min}}\}}^{k_{\text{max}}} c_l \pi_l \right\}\end{equation}
if $p=1$, where $c_l, \pi_l, k_{\text{min}}$ and $k_{\text{max}}$ are as in Lemma~\ref{lem:14}.
\end{thm}

\begin{proof}
In view of Theorem~\ref{thm:beta-eso-general}, we only need to bound $\|\hat{A}\|_{w^*,\hat{v}}^2$.
First, note that
\begin{equation}\label{eq:msms4323}
 \norm{\hat{A}}_{w^*, \hat{v}}^2  \quad \overset{\eqref{eq:12norms}}{=} \quad \max_{\|h\|_{w^*} = 1}(\|\hat{A} h\|_{\hat{v}}^*)^2 \quad \overset{\eqref{eq:norm2'}+\eqref{eq:09sddd}}{=} \quad  \max_{\norm{h}_{w^*}=1}\sum_{S \in \calP} p_S^{-1} (\|\hat{A}^Sh\|_v^*)^2.
\end{equation}
Further, it will be useful to observe that
  \begin{equation}\label{eq:091783h73}\hat{A}^S_{ji} \overset{\eqref{eq:A_ji}+\eqref{eq:jjd8d88d}}{=} p_S e_j^T A \sum_{k \in S}U_k U_k^T U_i \overset{\eqref{eq:U_iU_j}+ \eqref{eq:A_ji}+\eqref{eq:chi}}{=} p_S \chi_{(i \in S)}A_{ji} \enspace.\end{equation}
  For brevity, let us write $\eta_j \eqdef \abs{ \J(A^T e_j) \cap \hat{S}}$. As $\hat{S}$ is $\tau$-nice, adding dummy dependencies if necessary,
we can wlog assume that all rows of $A$ have the same number of nonzero blocks: $|\J(A^T e_j)| = \omega$ for all $j$. Thus,
$\pi_k \eqdef \Prob(\eta_j=k)$ does not depend on $j$.  Consider now two cases, depending on whether the norm $\|\cdot\|_v$ in $\R^m$ is defined with $p=1$ or $p=2$.

\begin{enumerate}

\item[(i)] For $p=2$ we can write
\begin{eqnarray}
 \norm{\hat{A}}_{w^*, \hat{v}}^2
 &\overset{\eqref{eq:msms4323} + \eqref{eq:q-norm} + \eqref{eq:nsj9292}}{=}&  \max_{\norm{h}_{w^*}=1}\sum_{S \in \calP} p_S^{-1} \sum_{j=1}^m v_j^{-2} \left({\sum_{i=1}^n \hat{A}^S_{ji} h^{(i)}}\right)^2 \notag \\
& \overset{\eqref{eq:091783h73}}{=} &  \max_{\norm{h}_{w^*}=1} \sum_{S \in \calP} p_S^{-1} \sum_{j=1}^m v_j^{-2}  \left(\sum_{i =1}^n p_S \chi_{(i \in S)}A_{ji}  h^{(i)}\right)^2 \notag \\
& \overset{\eqref{eq:subvector}}{=} & \max_{\norm{h}_{w^*}=1} \sum_{S \in \calP} p_S \sum_{j=1}^m v_j^{-2} \Big( \sum_{i \in  \J(A^T e_j) \cap S} A_{ji} h^{(i)}\Big)^2 \notag \\
& = & \max_{\norm{h}_{w^*}=1} \Exp\left[ \sum_{j=1}^m v_j^{-2} \Big( \sum_{i \in \J(A^T e_j) \cap \hat{S}} A_{ji} h^{(i)}\Big)^2 \right]\notag \\
& = & \max_{\norm{h}_{w^*}=1} \sum_{k=0}^n \Exp\left[ \left. \sum_{j=1}^m v_j^{-2} \Big( \sum_{i \in \J(A^T e_j) \cap \hat{S}} A_{ji} h^{(i)}\Big)^2 \right| \; \eta_j=k  \right] \pi_k \notag \\
& = & \max_{\norm{h}_{w^*}=1}  \sum_{k=0}^n  \sum_{j=1}^m v_j^{-2}\Exp\left[  \left. \Big( \sum_{i \in \J(A^T e_j) \cap \hat{S}} A_{ji} h^{(i)}\Big)^2 \right| \; \eta_j=k  \right] \pi_k. \label{eq:jska539036}
\end{eqnarray}

Using the Cauchy-Schwarz inequality, we can write
\begin{eqnarray}
\Exp\left[ \Big( \sum_{i \in \J(A^T e_j) \cap \hat{S}} A_{ji} h^{(i)}\Big)^2 \Big| \; \eta_j=k  \right]
& \overset{(\text{CS})}{\leq} &  \Exp\left[  \abs{\J(A^T e_j) \cap \hat{S}} \sum_{i \in \J(A^T e_j) \cap \hat{S}} (A_{ji}h^{(i)})^2 \Big| \; \eta_j=k \right] \notag \\
& = &  \Exp\left[ k  \sum_{i \in \J(A^T e_j) \cap \hat{S}} (A_{ji}h^{(i)})^2 \Big| \; \eta_j=k \right] \notag \\
& \overset{\eqref{eq:condexpJinterS}}{=} & \frac{k^2}{\omega}  \sum_{i \in \J(A^T e_j)} (A_{ji}  h^{(i)})^2.\label{eq:0153026jdjd}
\end{eqnarray}

Combining \eqref{eq:jska539036} and \eqref{eq:0153026jdjd}, we finally get
\begin{eqnarray*}
\|\hat{A}\|_{w^*,\hat{v}}^2 & \leq & \frac{1}{\omega} \sum_{k=0}^n k^2 \pi_k \left(\max_{\norm{h}_{w^*}=1}  \sum_{j=1}^mv_j^{-2} \sum_{i=1}^n
(A_{ji}  h^{(i)})^2 \right) \\
&  \overset{(\text{Lemma}~\ref{lem:666})}{\leq} & \frac{1}{\omega} \sum_{k=0}^n k^2 \pi_k  \quad  \overset{\eqref{eq:JinterS2}}{=} \quad \frac{\tau}{n} \Big( 1+\frac{(\omega-1)(\tau-1)}{\max(1, n-1)}\Big)  .
\end{eqnarray*}


 \item[(ii)] Consider now the case $p=1$.

\begin{eqnarray}
 \norm{\hat{A}}_{w^*,\hat{v}}^2 &\overset{\eqref{eq:msms4323} + \eqref{eq:q-norm} +\eqref{eq:nsj9292}}{=} &
 \max_{\norm{h}_{w^*}=1}\sum_{S \in \calP} p_S^{-1} \left[\max_{1\leq j\leq m} v_j^{-2} \left({\sum_{i=1}^n \hat{A}^S_{ji} h^{(i)}}\right)^2\right] \notag \\
 &\overset{\eqref{eq:091783h73}}{=} & \max_{\norm{h}_{w^*}=1}\sum_{S \in \calP} p_S^{-1} \left[\max_{1\leq j \leq m} v_j^{-2} \left({\sum_{i =1}^n  p_S \chi_{(i \in S)} A_{ji} h^{(i)}}\right)^2 \right] \notag\\
& \overset{\eqref{eq:subvector}}{=} & \max_{\norm{h}_{w^*}=1}\sum_{S \in \calP} p_S \left[\max_{1\leq j \leq m} v_j^{-2} \Big(\sum_{i \in \J(A^T e_j) \cap S} A_{ji} h^{(i)}\Big)^2 \right]\notag\\
& \overset{\text{(Cauchy-Schwarz)}}{\leq} & \max_{\norm{h}_{w^*}=1}\sum_{S \in \calP} p_S \left[\max_{1\leq j \leq m} v_j^{-2} |\J(A^T e_j) \cap S| \sum_{i \in \J(A^T e_j)\cap S} (A_{ji} h^{(i)})^2 \right]\notag\\
&\leq & \max_{\norm{h}_{w^*}=1}\sum_{S \in \calP} p_S \kappa_S \left[\max_{1\leq j \leq m} v_j^{-2}  \sum_{i \in S} (A_{ji} h^{(i)})^2 \right],\label{eq:kopa53636}
\end{eqnarray}
where  $\kappa_S \eqdef \max_{1\leq j\leq m} |\J(A^T e_j) \cap S|$. Consider the change of variables $y^{(i)} = (w_i^*)^{1/2} B_i^{1/2} h^{(i)}$. Utilizing essentially the same argument as in the proof of Lemma~\ref{lem:666} for $p=1$, we obtain
\begin{equation}\label{eq:8shshs}\max_{1\leq j \leq m} v_j^{-2}  \sum_{i \in S} (A_{ji} h^{(i)})^2 \leq \sum_{i \in S} \|y^{(i)}\|_E^2.\end{equation}
Since $\|y\|_E = \|h\|_{w^*}$, substituting \eqref{eq:8shshs} into \eqref{eq:kopa53636} gives
\begin{eqnarray}
 \norm{\hat{A}}_{w^*,\hat{v}}^2 & \leq & \max_{\norm{y}_{E}=1} \sum_{S \in \calP} p_S \kappa_S \sum_{i \in S}\|y^{(i)}\|_E^2 \quad = \quad \max_{\norm{y}_{E}=1} \sum_{i =1}^n \| y^{(i)}\|_E^2 \sum_{S \in \calP} p_S \kappa_S \chi_{(i \in S)}\notag\\
& = & \max_{\norm{y}_{E}=1} \sum_{i =1}^n \| y^{(i)}\|_E^2 \Exp[ \kappa_{\hat{S}} \chi_{(i \in \hat{S})}] \notag\\
& = & \max_{1 \leq i \leq n}  \Exp[\kappa_{\hat{S}} \chi_{(i \in \hat{S})}] \quad = \quad  \max_{1 \leq i \leq n}  \Exp[\max_{1\leq j\leq m} |\J(A^T e_j) \cap \hat{S}|\times \chi_{(i \in \hat{S})}].\label{eq:jsjs09nsd5sdj}
\end{eqnarray}
It now only remains to apply inequality \eqref{eq:98ss54sgs} used with $J_j = \J(A^T e_j)$.
\end{enumerate}

\end{proof}

\begin{figure}

\centering

\includegraphics[height=35ex]{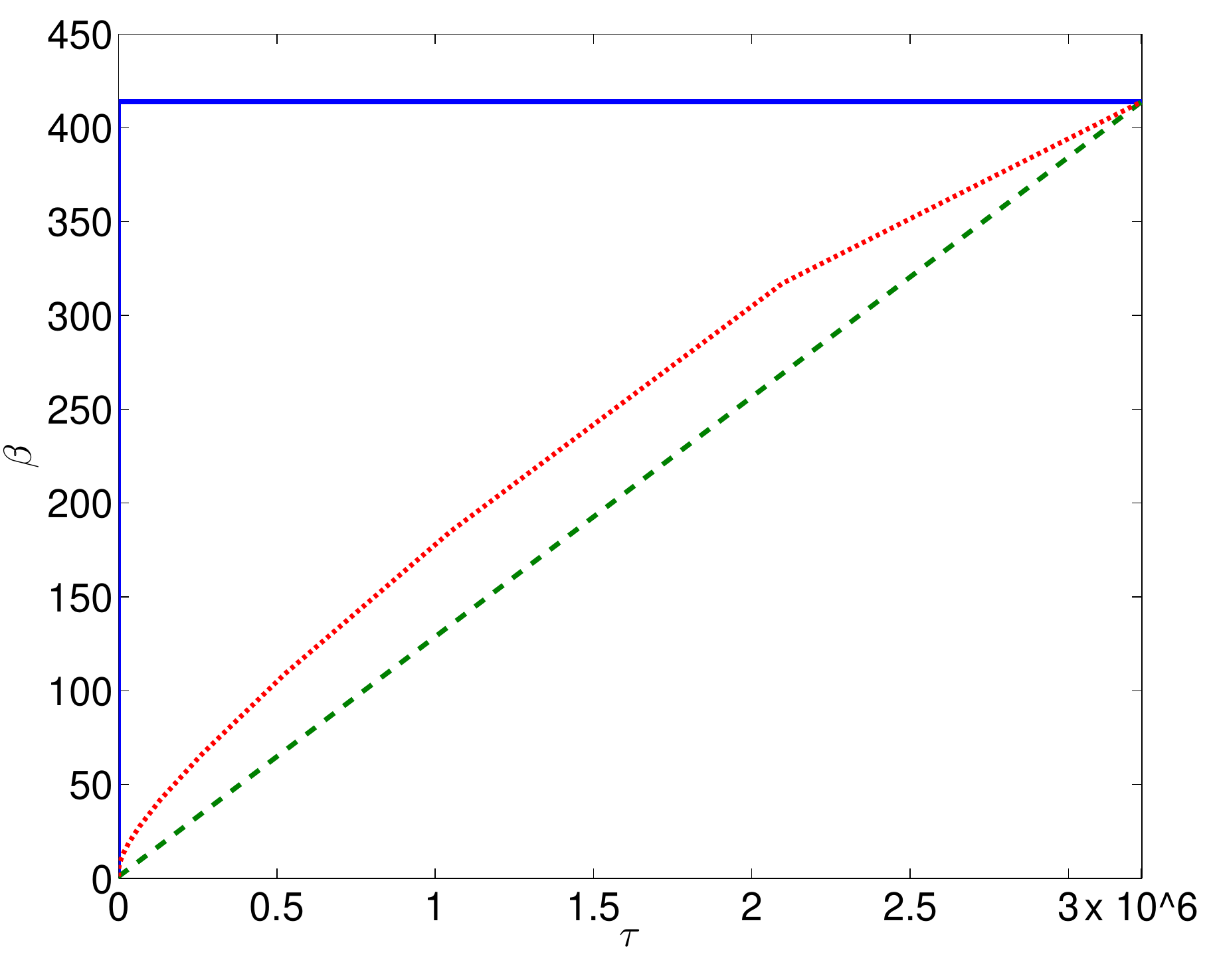}
\quad
\includegraphics[height=35ex]{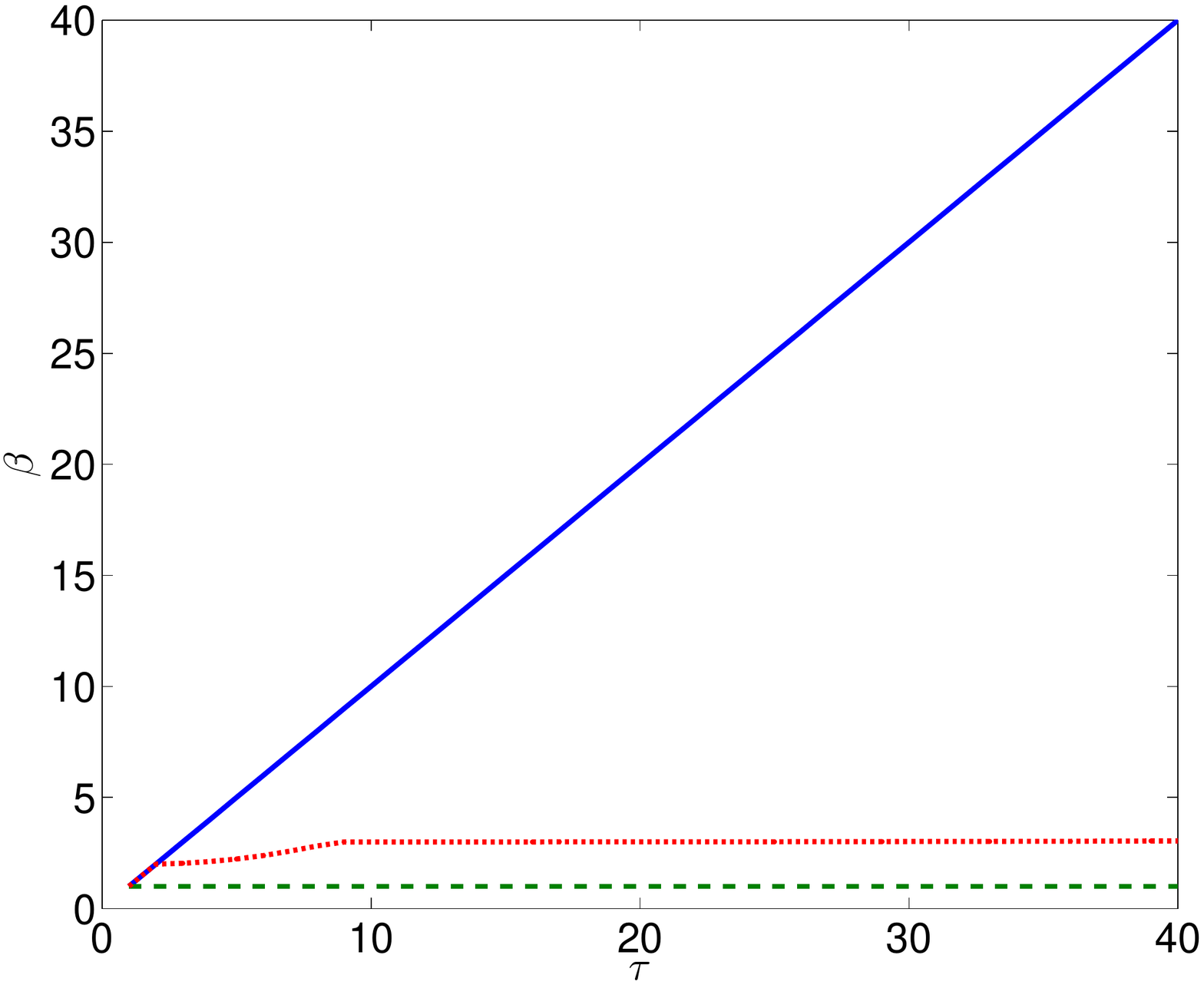}

\caption{
Comparison of the three formulae for $\beta'$ as a function of the number of processors~$\tau$ (smaller $\beta'$ is better). We have used matrix $A \in \R^{m \times n}$ with $m = 2,396,130$, $n = 3,231,961$ and $\omega = 414$.  \textbf{Blue solid line:} $\tau$-uniform sampling, $\beta'_1=\min\{\omega, \tau\}$ (Theorem~\ref{thm:beta-dso}). \textbf{Green dashed line:}  $\tau$-nice sampling and $p=2$, $\beta'_2=1+\frac{(\omega-1)(\tau-1)}{\max\{1,n-1\}}$ (Theorem \ref{thm:beta-eso-taunice}). \textbf{Red dash-dotted line:} $\tau$-nice sampling and $p=1$, $\beta'_3$ follows \eqref{eq:p=1beta} in Theorem~\ref{thm:beta-eso-taunice}. Note that $\beta'_1$ reaches its maximal value  $\omega$ quickly, whereas  $\beta'_2$ increases slowly. When $\tau$ is small compared to $n$, this means that $\beta'_2$ remains close to 1. As shown in Section~\ref{sec:complexity} (see Theorems~\ref{thm:complexity_strong_convexity_nonsmooth} and \ref{thm:complexity_strong_convexity_smoothed}), small values of $\beta'$ directly translate into better complexity and parallelization speedup. \textbf{Left:} Large number of processors. \textbf{Right:} Zoom for smaller number of processors.
}

\label{fig:compeso}
\end{figure}

Let us now comment on some aspects of the above result.
\begin{enumerate}
\item It is possible to draw a link between $\beta'_2$ and $\beta'_3$. In view of \eqref{eq:JinterS2}, for $p=2$ we have \[ \beta'_2 = \tfrac{n}{\tau}\max_{1\leq i \leq n} \Exp[|\J(A^T e_j) \cap \hat{S}| \times \chi_{(i \in \hat{S})}],\]
where $j$ is such that  $|\J(A^T e_j)| = \omega$ (we can wlog assume this holds for all $j$). On the other hand, as is apparent from \eqref{eq:jsjs09nsd5sdj}, for $p=1$ we can replace $\beta'_3$ by
\[ \beta''_3 \eqdef \tfrac{n}{\tau}\max_{1\leq i \leq n} \Exp[\max_{1\leq j \leq m}|\J(A^T e_j) \cap \hat{S}| \times \chi_{(i \in \hat{S})}].\]
Clearly, $\beta'_2\leq \beta''_3 \leq m \beta_2'$. However, in many situations,  $\beta_3'' \approx \beta_2'$ (see Figure~\ref{fig:compeso}). Recall that a small $\beta$ is good for Algorithm~\ref{alg:pcdm1} (this will be formally proved in the next section).

\item If we let $\omega^* = \max_i \{j \st A_{ji} \neq 0\}$ (maximum number of nonzero rows in matrices $A_1,\dots, A_n$), then in the $p=1$ case we can replace $\beta_3'$ by the smaller quantity
\[
\beta_3''' \eqdef \frac{\tau}{n} \sum_{k=k_{\text{min}}}^{k_{\text{max}}} \min \left\{ 1,  \sum_{l=k}^{k_{\text{max}}}
\left( m\frac{n}{n-\omega}\frac{\tau-l}{\tau} + n\frac{\omega^*}{\omega}\frac{l}{\tau} \right) \pi_l \right\}.
\]

\end{enumerate}

\section{Iteration Complexity} \label{sec:complexity}

In this section we formulate \emph{concrete} complexity results for Algorithm~\ref{alg:pcdm1} applied to problem \eqref{eq:P} by combining the generic results proved in \cite{RT:PCDM} and outlined in the introduction, Lemma~\ref{eq:lemma098098sdsd} (which draws a link between \eqref{eq:P} and \eqref{eq:Psmooth} and, \emph{most importantly}, the concrete values of $\beta$ and $w$ established in this paper for Nesterov separable functions and $\tau$-uniform and $\tau$-nice samplings.

A function $\phi: \R^N \to \R\cup \{+\infty\}$ is strongly convex with respect to the norm $\|\cdot\|_w$ with convexity parameter $\sigma_\phi(w)\geq 0$ if for all $x,\bar{x} \in \dom \phi$,
\[\phi(x) \geq \phi(\bar{x}) + \ve{\phi'(\bar{x})}{x-\bar{x}} + \frac{\sigma_{\phi}(w)}{2}\|x-\bar{x}\|_w^2,\]
where $\phi'(\bar{x})$ is any subgradient of $\phi$ at $\bar{x}$.

For $x_0 \in \R^N$ we let ${\cal L}^\delta_{\mu}(x_0) \eqdef \{x \st \FF_\mu(x)\leq \FF_\mu(x_0) + \delta\}$ and let
\[{\cal D}^\delta_{w,\mu}(x_0) \eqdef  \max_{x,y} \{\|x-y\|_w \st x, y \in {\cal L}^\delta_\mu(x_0)\}\]
be the diameter of this set in the norm $\|\cdot\|_w$.

It will be useful to recall some basic notation from Section~\ref{sec:SPCDM} that the theorems of this section will refer to: $\FF(x) = \ff(x)+\Psi(x)$, and  $\FF_\mu(x) = \ff_\mu(x)+\Psi(x)$, with
\[\ff(x) = \max_{z \in Q}\{\ve{Ax}{z}-g(z)\}, \quad \ff_\mu(x) = \max_{z\in Q} \{\ve{Ax}{z} - g(z) - \mu d(z)\},\] where $d$ is a prox function on $Q$ (it is strongly convex on $Q$  wrt $\|\cdot\|_v$, with constant $\sigma$) and $D = \max_{x \in Q} d(x)$. Recall also that $\|\cdot\|_v$ is a weighted $p$ norm on $\R^m$, with weights $v_1,\dots,v_m>0$. Also recall that $\|\cdot\|_w$ is a norm defined as a weighted quadratic mean of the block-norms $\ve{B_i x^{(i)}}{x^{(i)}}^{1/2}$, with weights $w_1,\dots,w_n>0$.

\begin{thm}[Complexity: smoothed composite problem \eqref{eq:Psmooth}] \label{thm:complexity_strong_convexity_smoothed}Pick $x_0 \in \dom \Psi$ and  let $\{x_k\}_{k\geq 0}$ be the sequence of random iterates produced by the smoothed parallel descent method (Algorithm~\ref{alg:pcdm1}) with the following setup:
\begin{itemize}
\item[(i)] $\{S_k\}_{k\geq 0}$ is an iid sequence of $\tau$-uniform samplings, where $\tau \in \{1,2,\dots,n\}$,
\item[(ii)] $w = w^*$, where $w^*$ is defined in \eqref{eq:w_i},
\item[(iii)] $\beta = \tfrac{\beta'}{\sigma \mu}$, where $\beta' = 1+\tfrac{(\omega-1)(\tau-1)}{\max\{1,n-1\}}$ if the samplings are $\tau$-nice and $p=2$,
 $\beta'$ is given by \eqref{eq:p=1beta} if the samplings are $\tau$-nice and $p=1$, and $\beta' = \min\{\omega,\tau\}$ if the samplings are not $\tau$-nice ($\omega$ is the degree of Nesterov separability).
\end{itemize}
Choose error tolerance $0 < \epsilon < \FF_\mu(x_0)-\min_x \FF_\mu(x)$, confidence level  $0<\rho < 1$ and iteration counter $k$ as follows:
\begin{enumerate}
\item[(i)]
if $\FF_\mu$ is strongly convex with $\sigma_{\ff_\mu}(w^*) + \sigma_{\Psi}(w^*) >0$, choose \[k \geq \frac{n}{\tau} \times  \frac{\tfrac{\beta'}{\mu \sigma}+\sigma_\Psi(w^*)}{\sigma_{f_\mu}(w^*) + \sigma_\Psi(w^*)} \times \log \left(\frac{\FF_\mu(x_0)-\min_x \FF_\mu(x)}{\epsilon \rho}\right),\]
\item[(ii)] otherwise additionally assume\footnote{This assumption is not restrictive as $\beta' \geq 1$, $n \geq \tau$ and $\mu, \sigma$ are usually small. However, it is technically needed.} that $\epsilon < \tfrac{2 n\beta}{\tau }$ and that\footnote{Instead of the assumption $\beta'=\min\{\omega,\tau\}$ it suffices to include an additional step into SPCDM which accepts only updates decreasing the loss. That is, $x_{k+1}$ is set $x_k$ in case $\FF_\mu(x_{k+1})>\FF_\mu(x_k)$. However, function evaluation is not recommended as it would considerable slow down the method. In our experiments we have never encountered a problem with using the more efficient $\tau$-nice sampling even in the non-strongly convex case. In fact, this assumption may just be an artifact of the analysis.}  $\beta'=\min\{\omega,\tau\}$, and choose
\[k \geq \frac{n\beta'}{\tau} \times \frac{2 ({\cal D}^0_{w^*,\mu}(x_0))^2}{ \mu \sigma\epsilon} \times \log\left(\frac{\FF_\mu(x_0) - \min_x \FF_\mu(x)}{\epsilon \rho}\right).\]
\end{enumerate}
Then
\[\Prob(\FF_\mu(x_k) - \min_x \FF_\mu(x) \leq \epsilon)\geq 1-\rho.\]
\end{thm}

\begin{proof}
This follows from the generic complexity bounds proved by Richt\'{a}rik and Tak\'{a}\v{c} \cite[Theorem~19(ii) and Theorem~20]{RT:PCDM} and Theorems~\ref{thm:beta-dso} and \ref{thm:beta-eso-taunice} giving formulas for $\beta'$ and $w^*$ for which $(\ff_\mu,\hat{S})\sim \ESO(\tfrac{\beta'}{\sigma \mu},w^*)$.
\end{proof}

We now we consider solving the nonsmooth problem \eqref{eq:P} by applying Algorithm~\ref{alg:pcdm1} to its smooth approximation  \eqref{eq:Psmooth} for a specific value of the smoothing parameter $\mu$.

\begin{thm}[Complexity: nonsmooth composite problem \eqref{eq:P}] \label{thm:complexity_strong_convexity_nonsmooth}Pick $x_0 \in \dom \Psi$ and  let $\{x_k\}_{k\geq 0}$ be the sequence of random iterates produced by the smoothed parallel descent method (Algorithm~\ref{alg:pcdm1}) with the same setup as in Theorem~\ref{thm:complexity_strong_convexity_smoothed}, where
$\mu=\tfrac{\epsilon'}{2D}$ and $0 < \epsilon' < \FF(x_0)-\min_x \FF(x)$. Further, choose confidence level  $0<\rho < 1$ and iteration counter as follows:
\begin{enumerate}
\item[(i)]
if $\FF_\mu$ is strongly convex with $\sigma_{\ff_\mu}(w^*) + \sigma_{\Psi}(w^*) >0$, choose \[k \geq \frac{n}{\tau} \times  \frac{\frac{2\beta' D}{\sigma \epsilon'}+\sigma_\Psi(w^*)}{\sigma_{f_\mu}(w^*) + \sigma_\Psi(w^*)} \times \log \left(\frac{2(\FF(x_0)-\min_x \FF(x)) + \epsilon'}{\epsilon' \rho}\right),\]
\item[(ii)] otherwise additionally assume that $(\epsilon')^2 < \tfrac{8nD\beta'}{ \sigma \tau}$ and that $\beta'=\min\{\omega,\tau\}$, and choose
\[k \geq \frac{n\beta'}{\tau} \times \frac{8D    ({\cal D}_{w^*,0}^{\epsilon'/2}(x_0))^2}{\sigma (\epsilon')^2} \times \log\left(\frac{2(\FF(x_0) - \min_x \FF(x))+\epsilon'}{\epsilon' \rho}\right).\]
\end{enumerate}
Then
\[\Prob(\FF(x_k) - \min_x \FF(x) \leq \epsilon')\geq 1-\rho.\]
\end{thm}

%
\begin{proof}  We will apply Theorem~\ref{thm:complexity_strong_convexity_smoothed} with $\epsilon = \tfrac{\epsilon'}{2}$ and $\mu = \tfrac{\epsilon'}{2D}$. All that we need to argue in case (i) (and we need this in case (ii) as well) is: (a) $\epsilon < \FF_\mu(x_0) - \FF_\mu(x_\mu^*)$, where $x_\mu^*=\arg\min_x \FF_\mu(x)$ (this is needed to satisfy the assumption about $\epsilon$), (b) $\FF_\mu(x_0) - \FF_\mu(x_\mu^*) \leq \FF(x_0) - \FF(x^*) + \epsilon$ (this is needed for logarithmic factor in the iteration counter) and (c) $\Prob(\FF_\mu(x_k)- F_\mu(x_\mu^*)\leq \epsilon)\leq \Prob(\FF(x_k)-\FF(x^*) \leq \epsilon')$, where $x^* =\arg\min_x \FF(x)$. Inequality (a) follows  by combining our assumption with Lemma~\ref{eq:lemma098098sdsd}. Indeed, the assumption $\epsilon' < \FF(x_0) -\FF(x^*)$ can be written as $\tfrac{\epsilon'}{2}<\FF(x_0)- \FF(x^*) - \mu D$, which combined with the second inequality  in \eqref{eq:propsjhs88sd}, used with $x=x_0$, yields the result. Further, (b) is identical to the first inequality in Lemma~\ref{eq:lemma098098sdsd} used with $x=x_0$. Finally, (c) holds since  the second inequality of Lemma~\ref{eq:lemma098098sdsd} with $x = x_k$ says that  $\FF_\mu(x_k)- F_\mu(x_\mu^*) \leq \tfrac{\epsilon'}{2}$ implies $\FF(x_k)-\FF(x^*) \leq \epsilon'$.

In case (ii) we additionally need to argue that: (d) $\epsilon < \tfrac{2n \beta}{\tau}$ and (e) ${\cal D}_{w^*,\mu}^0(x_0) \leq {\cal D}_{w^*,0}^{\epsilon'/2}(x_0)$. Note that (d) is equivalent to the assumption $(\epsilon')^2 < \tfrac{8nD \beta'}{\sigma \tau}$. Notice that as long as $\FF_\mu(x)\leq \FF_\mu(x_0)$, we have \[\FF(x) \overset{\eqref{eq:eps09809}}{\leq} \FF_\mu(x) + \tfrac{\epsilon'}{2} \leq \FF_\mu(x_0) + \tfrac{\epsilon'}{2} \overset{\eqref{eq:eps09809}}{\leq}  \FF(x_0) + \tfrac{\epsilon'}{2},\]
and hence ${\cal L}^0_{\mu}(x_0) \subset {\cal L}^{\epsilon'/2}_{0}(x_0)$, which implies (e).
\end{proof}

Let us now briefly comment on the results.
\begin{enumerate}
\item If we choose the separable regularizer $\Psi(x) = \tfrac{\delta}{2}\|x\|_{w^*}^2$, then $\sigma_{\Psi}(w^*)= \delta$ and the strong convexity assumption is satisfied, irrespective of whether $\ff_\mu$ is strongly convex or not. A regularizer of this type is often chosen in machine learning applications. 

\item Theorem~\ref{thm:complexity_strong_convexity_nonsmooth} covers the problem $\min \FF(x)$ and hence we have on purpose formulated the results without any reference to the smoothed problem (with the exception of dependence on $\sigma_{f_\mu}(w^*)$ in case (i)). We traded a (very) minor loss in the quality of the results for a more direct formulation.

\item As the confidence level is inside a logarithm, it is easy to obtain a high probability result with this randomized algorithm.
For problem \eqref{eq:P} in the non-strongly convex case,  iteration complexity is $O((\epsilon')^{-2})$ (ignoring the logarithmic term), which is comparable to other techniques available for the minimization of nonsmooth convex functions such as the subgradient method. In the strongly convex case the dependence is $O((\epsilon')^{-1})$. Note, however, that in many applications solutions only of moderate or low accuracy are required, and the focus is on the dependence on the number of processors $\tau$ instead. In this regard, our methods have excellent theoretical parallelization speedup properties.

\item It is clear from the complexity results that as more processors $\tau$ are used, the method requires fewer iterations, and the speedup gets higher for smaller values of $\omega$ (the degree of Nesterov separability of $\ff$). However, the situation is even better if the regularized $\Psi$ is strongly convex -- the degree of Nesterov separability then has a weaker effect on slowing down parallelization speedup.

\item For $\tau$-nice samplings, $\beta$ changes depending on $p$ (the type of dual norm $\|\cdot\|_v$). However, $\sigma$ changes also, as this is the strong convexity constant of the prox function $d$ with respect to the dual norm $\|\cdot\|_v$.

\end{enumerate}

\section{Computational Experiments} \label{sec:applications}

In this section we consider the application of the smoothed parallel coordinate descent method (SPCDM) to three special problems and comment on some \emph{preliminary} computational experiments. For simplicity, in all examples we assume all blocks are of size 1 ($N_i=1$ for all $i$) and  fix $\Psi\equiv 0$.

In all tests we used a shared-memory workstation with 32 Intel Xeon processors at 2.6~GHz and 128~GB RAM. We coded an asynchronous version of SPCDM to limit communication costs and approximated $\tau$-nice sampling by a $\tau$-independent sampling as in~\cite{RT:PCDM} (the latter is very easy to generate in parallel). 

\subsection{L-infinity regression / linear programming}
Here we consider the the problem of minimizing the function
\[
 \ff(x) =  \|\tilde{A} x-\tilde{b}\|_\infty =\max_{u \in Q} \{\ve{Ax}{ u} - \ve{b}{u}\},
\]
where \[\tilde{A} \in \R^{m \times n}, \qquad \tilde{b} \in \R^m,\qquad  A=\left[\begin{smallmatrix} \tilde{A} \\ -\tilde{A} \end{smallmatrix}\right] \in \R^{2m \times n}, \qquad b=\left[\begin{smallmatrix} \tilde{b} \\ -\tilde{b} \end{smallmatrix}\right] \in \R^{2m}\] and $Q \eqdef \{u_j \in \R^{2m}\st \sum_j u_j =1,\; u_j\geq 0\}$ is the unit simplex in $\R^{2m}$. We choose the dual norm $\|\cdot\|_v$ in $\R^{2m}$ with $p=1$ and $v_j=1$ for all $j$. Further, we choose the prox function
$d(u)= \log(2m) + \sum_{j=1}^{2m} u_j \log(u_j)$ with center $u_0 = (1,1,\dots,1)/(2m)$. It can be shown that $\sigma=1$ and $D= \log(2m)$. Moreover, we let all blocks be of size 1 ($N_i=1$), choose $B_i=1$ for all $i$ in the definition of the primal norm and
\[
 w_i^* \overset{\eqref{eq:w_i^*-simple}}{=} \max_{1 \leq j \leq 2m} A_{ji}^2  =  \max_{1\leq j\leq m} \tilde{A}_{ji}^2 \enspace.
\]
The smooth approximation of $\ff$ is given by
\begin{equation}\label{eq:9s83n8nd}
 \ff_\mu(x)=\mu \log \left(  \frac{1}{2m} \sum_{j=1}^{2m} \exp\left( \frac{e_j^T A x - b_j}{\mu}\right) \right) \enspace.
\end{equation}

\textbf{Experiment.}  In this experiment we minimize $\ff_\mu$ utilizing $\tau$-nice sampling and parameter $\beta$ given by
\eqref{eq:p=1beta}. We first compare SPCDM (Algorithm~\ref{alg:pcdm1}) with several other methods, see Table~\ref{tab:compare}.

We perform a small scale experiment so that we can solve the problem directly as a linear program with GLPK. The simplex method struggles to progress initially but eventually finds the exact solution quickly. The accelerated gradient algorithm of Nesterov is easily parallelizable, which makes it competitive, but it suffers from small stepsizes (we chose here the estimate for the Lipschitz constant of the gradient given in~\cite{Nesterov-Subgrad-Huge} for this problem). A very efficient algorithm for the minimization of the infinity norm
is Nesterov's sparse subgradient method \cite{Nesterov-Subgrad-Huge} that is the fastest in our tests even when it uses a single core only. It performs full subgradient iterations in a very cheap way, utilizing the fact that the subgradients are sparse. The method has a sublinear in $n$ complexity.
However, in order for the method to take long steps, one needs to know the optimal value in advance. Otherwise, the algorithm is much slower, as is shown in the table.

\begin{centering}

\begin{table}[!ht]
 \begin{tabular}{|l|r|r|}

\hline
 {\em Algorithm} & {\em \# iterations} & {\em time (second) }\\
\hline
GLPK's simplex & 55,899 ~~~ &  681   ~~~ \\
\hline
Accelerated gradient~\cite{Nesterov05:smooth}, $\tau=16$ cores & 8,563  ~~~ & 246  ~~~  \\
\hline
Sparse subgradient~\cite{Nesterov-Subgrad-Huge}, optimal value known & 1,730  ~~~ & 6.4 ~~ \\
\hline
Sparse subgradient~\cite{Nesterov-Subgrad-Huge}, optimal value unknown & 166,686  ~~~ & 544  ~~~ \\
\hline
Smoothed PCDM (Theorem~\ref{thm:complexity_strong_convexity_smoothed}), $\tau = 4$ cores ($\beta=3.0$) & 15,700,000  ~~~ & 53   ~~~ \\
\hline
Smoothed PCDM (Theorem~\ref{thm:complexity_strong_convexity_smoothed}), $\tau =16$ cores ($\beta=5.4$) & 7,000,000  ~~~ & 37  ~~~  \\
\hline
 \end{tabular}
\caption{Comparison of various algorithms for the minimization of $\ff(x)=\|\tilde{A}x-\tilde{b}\|_\infty$, where $\tilde{A}$ and $\tilde{b}$ are
taken from the Dorothea dataset~\cite{guyon2004result} ($m=800$, $n=100,000$, $\omega=6,061$) and $\epsilon=0.01$.}
\label{tab:compare}
\end{table}

\end{centering}

For this problem, and without the knowledge of the optimal value, the smoothed parallel coordinate descent method presented in this paper
is the fastest algorithm. Many iterations are needed but they are very cheap: in its serial version, at each iteration one only needs to compute one partial derivative and to update 1 coordinate of the optimization variable, the residuals and the normalization factor. The worst case algorithmic complexity of one iteration is thus proportional to the number of nonzero elements in one column; on average.

Observe that quadrupling the number of cores does not divide by 4 the computational time because of the increase in the $\beta$ parameter. Also note that we have tested our method using the parameters dictated by the theory. In our experience the performance of the method imroves for a smaller value of $\beta$: this leads to larger stepsizes and the method often tolerates this.

%
%
%
%
%
%
%
%
%

\bigskip
{\footnotesize \emph{Remark:} There are numerical issues with the smooth approximation of the infinity norm because it involves exponentials of potentially large numbers.  A safe way of computing \eqref{eq:9s83n8nd}
is to compute first $\bar{r}=\max_{1 \leq j \leq 2m} (Ax-b)_j$ and to use the safe formula
\[
 \ff_\mu(x)=\bar{r} + \mu \log \left( \frac{1}{2m} \sum_{j=1}^{2m} \exp\left(\frac{(Ax-b)_j-\bar{r}}{\mu}\right) \right).
\]
However, this formula is not suitable for parallel updates because the logarithm prevents us from making reductions.
We adapted it in the following way to deal with parallel updates. Suppose we have already computed $\ff_\mu(x)$. Then
$\ff_\mu(x+h)= \ff_\mu(x)+   \mu \log \left(  S_x(\delta) \right)$, where
\[
 S_x(h) \eqdef \frac{1}{2m}\sum_{j=1}^{2m} \exp\left(\frac{(A x- b)_j+(Ah)_j-\ff_\mu(x)}{\mu}\right)
\]
In particular, $S_x(0)=1$. Thus, as long as the updates are reasonably small, one can compute $\exp[((Ax-b)_j+(Ah)_j-\ff_\mu(x))/\mu]$
and update the sum in parallel. From time to time (for instance every $n$ iterations or when $S_x$ becomes small), we recompute $\ff_\mu(x)$ from scratch and reset $h$ to zero.
}
\bigskip

\subsection{L1 regression}

Here we consider the problem of minimizing the function
\[
 \ff(x) = \|Ax-b\|_1 = \max_{u \in Q}  \{\ve{Ax}{u}-\ve{b}{u}\},
\]
where $Q = [-1,1]^n$. We define the dual norm $\|\cdot\|_v$ with $p=2$ and $v_j = \sum_{i=1}^n A_{ji}^2$ for all  $j=1,2,\dots,m$. Further, we choose the prox function
$d(z)=\tfrac{1}{2} \norm{z}_v^2$ with center $z_0=0$. Clearly,  $\sigma=1$ and $D=\tfrac{1}{2}\sum_{j=1}^m v_j= \sum_{j=1}^m \sum_{i=1}^n A_{ji}^2 = \|A\|_F^2$. Moreover, we choose $B_i=1$ for all $i=1,2,\dots,n$ in the definition of the primal norm and

\[
 w_i^* \overset{\eqref{eq:w_i^*-simple}}{=} \sum_{j=1}^m v_j^{-2} A_{ji}^2, \qquad i=1,2,\dots,n.
\]
The smooth approximation of $\ff$ is given by
\[
 \ff_\mu(x)=\sum_{j=1}^m \norm{e_j^T A}_{w^*}^* \psi_\mu \left( \frac{\abs{e_j^T A x - b_j}}{ \norm{e_j^T A}_v^*} \right), \qquad \psi_\mu(t)=\begin{cases}  \frac{t^2}{2\mu}, & 0 \leq t \leq \mu, \\ t - \frac{\mu}{2}, & \mu \leq t. \end{cases}
\]

\begin{figure}[htpb]
\centering
 \includegraphics[width=21em]{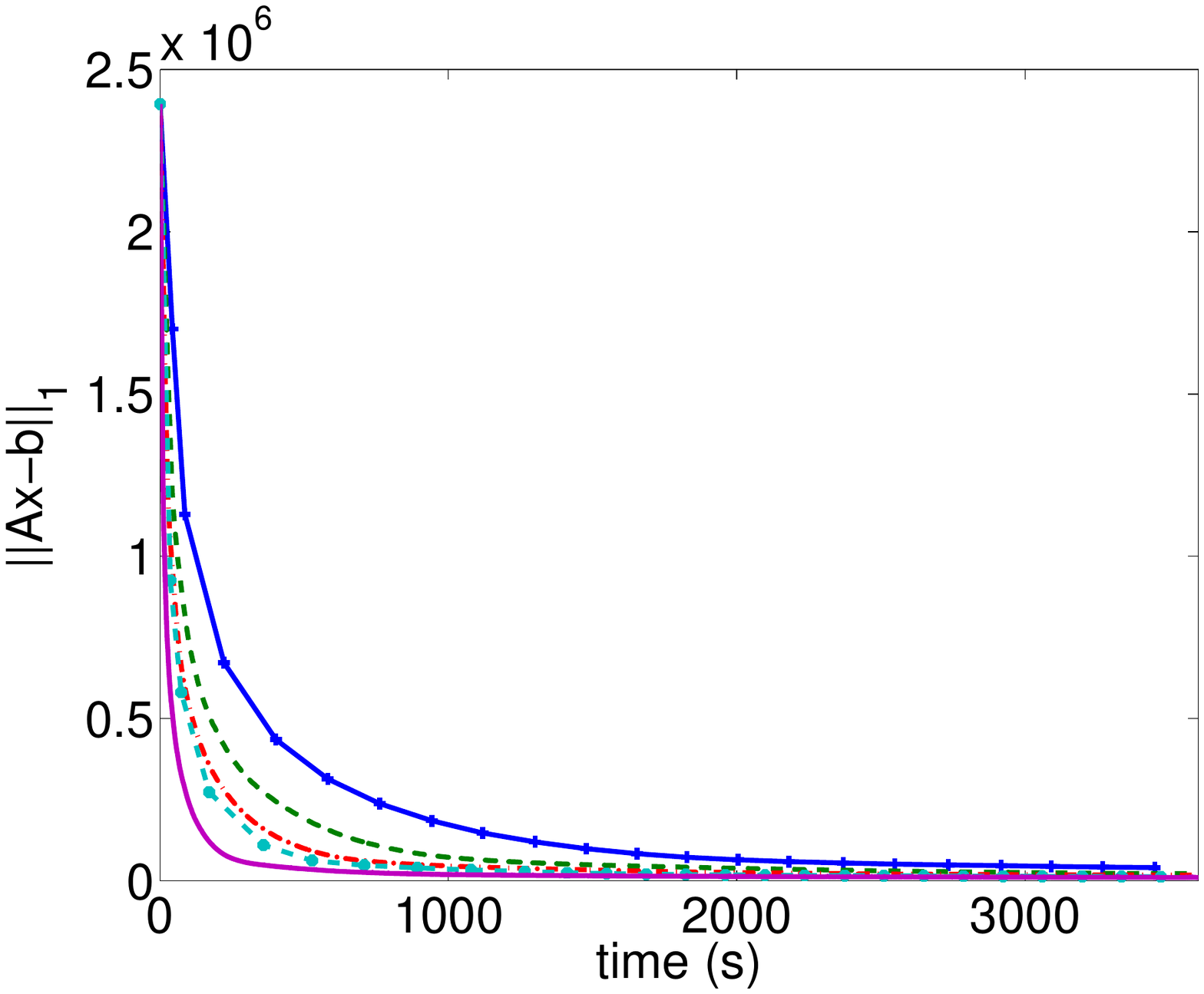}
 \includegraphics[width=21em]{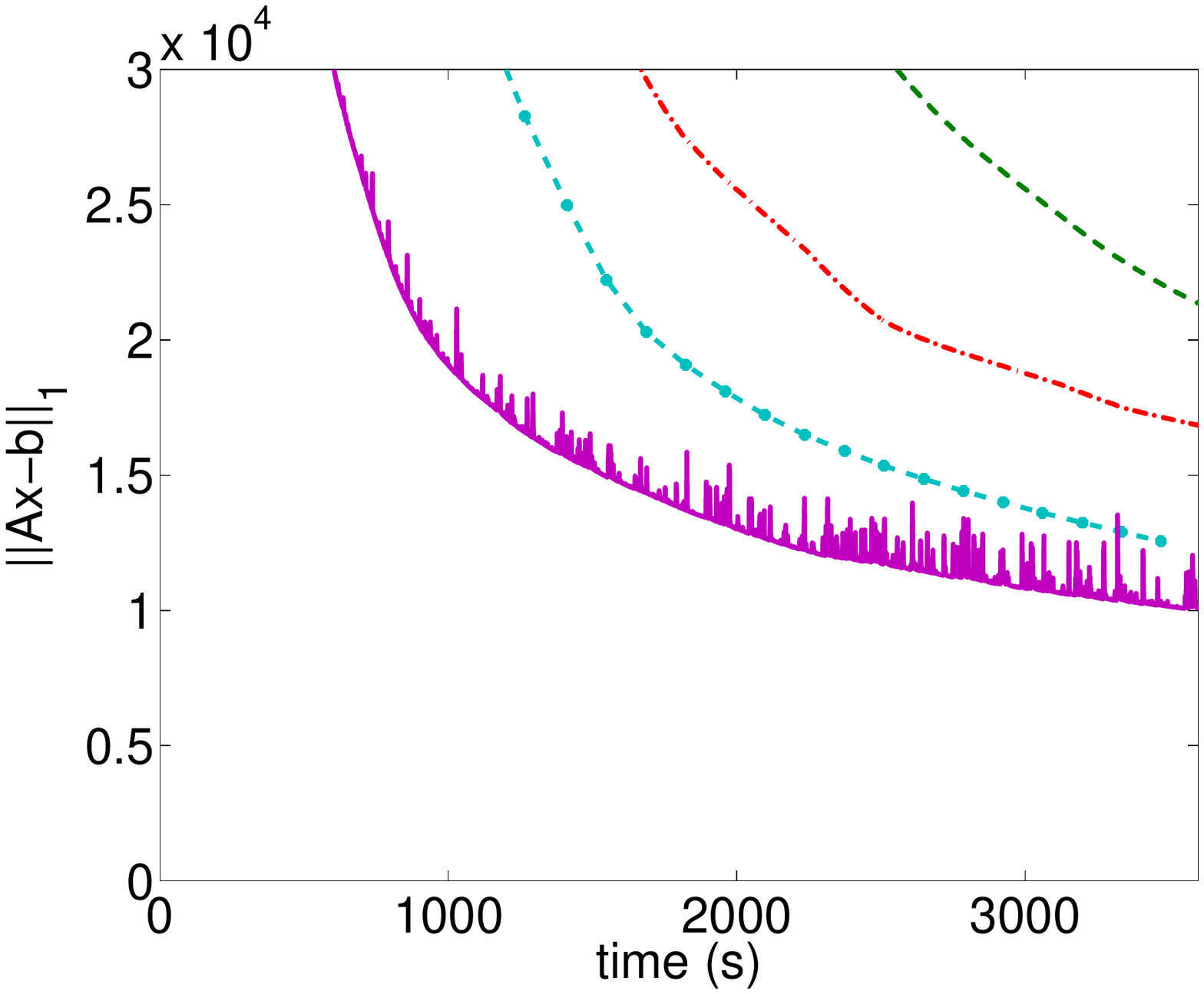}
\caption{Performance of SPCDM on the problem of minimizing $\ff(x)=\|Ax-b\|_1$ where $A$ and $b$ are
given by the URL reputation dataset. We have run the method until the function value was decreased by a factor of $240$.
\textbf{blue solid line with crosses:} $\tau=1$; \textbf{green dashed line:} $\tau=2$;
\textbf{red dash-dotted line:} $\tau=4$; \textbf{cyan dashed line with stars:} $\tau=8$;
\textbf{solid purple line:} $\tau=16$. \textbf{Left:} Decrease of the objective value in time.    We can see that parallelization speedup is
proportional to the number of processors. \textbf{Right:} Zoom on smaller objective values. We can see that
the algorithm is not monotonic but monotonic on average.}
\label{fig:expe1norm}
\end{figure}

\bigskip

{\footnotesize \emph{Remark:} Note that in \cite{Nesterov05:smooth}, the dual norm is defined from the primal norm.
In the present work, we need to define the dual norm first since otherwise
the definitions of the norms would cycle. However, the definitions above
give the choice of $v$ that minimizes the term
\[
D\|e\|_{w^*}^2 = D \sum_{i=1}^n w_i^* = (\sum_{j=1}^m v_j) (\sum_{j'=1}^m v_{j'}^{-2} A_{j'i}^2),
\]
where $e = (1,1,\dots,1)\in \R^N$. We believe that in the non-strongly convex case one can replace in the complexity estimates the squared diameter of the level set by $\|x_0-x^*\|_{w^*}^2$, which would then mean that a product of the form $D\|x_0-x_*\|_{w^*}^2$ appears in the complexity. The above choice of the weights $v_1,\dots,v_m$ minimizes this product under assuming that $x_0-x^*$ is proportional to $e$.
}

\bigskip

\textbf{Experiment.}  We performed our medium scale numerical experiments (in the case of L1 regression and exponential loss minimization (Section~\ref{sec:adaboost})) on the URL reputation dataset \cite{ma2009identifying}.
It gathers $n=3,231,961$ features about $m=4,792,260$ URLs collected during 120 days. The feature matrix is sparse but it has some dense columns.
The maximum number of nonzero elements in a row is $\omega=414$. The vector of labels classifies the page as spam or not.

We applied SPCDM with $\tau$-nice sampling, following the setup described in Theorem~\ref{thm:complexity_strong_convexity_nonsmooth}. The results for $\ff(x) = \|Ax-b\|_1$ are gathered in Figure~\ref{fig:expe1norm}.  We can see that parallelization speedup is
proportional to the number of processors. In the right plot we observe that
the algorithm is not monotonic but monotonic on average.

%
%

\subsection{Logarithm of the exponential loss}\label{sec:adaboost}

Here we consider the problem of minimizing the function \begin{equation}\label{eq:log_exp_loss}\ff_1(x)=\log \left(\frac{1}{m}\sum_{j=1}^m \exp(b_j(Ax)_j)\right).\end{equation}

The AdaBoost algorithm~\cite{freund1995decision} minimizes the exponential loss $\exp(\ff_1(x))$
by a greedy serial coordinate descent method (i.e., at each iteration, one selects the coordinate
corresponding to the largest directional derivative and updates that coordinate only). Here we observe that  $\ff_1$ is Nesterov separable as it is the smooth approximation of
\[
 \ff(x)=\max_{1 \leq j \leq m} b_j (Ax)_j
\]
with $\mu=1$. Hence, we can minimize $\ff_1$  by parallel coordinate descent with $\tau$-nice sampling and $\beta$ given by \eqref{beta_1}.

Convergence of AdaBoost is not a trivial result because the
minimizing sequences may be unbounded. The proof relies on a decomposition of the
optimization variables to an unbounded part and a bounded part~\cite{mukherjee2011rate, telgarsky2012primal}.
The original result gives iteration complexity $O(\frac{1}{\epsilon})$.

Parallel versions of AdaBoost have previously been studied. I our notation, Collins, Shapire and Singer~\cite{collins2002logistic} use $\tau=n$ and $\beta=\omega$. Palit and Reddy~\cite{palit2012scalable} use a generalized greedy sampling
and take $\beta=\tau$ (number of processors). In the present work, we use randomized samplings and
we can take $\beta\ll \min\{\omega,\tau\}$ with the $\tau$-nice sampling. As discussed before, this value of $\beta$ can be $O(\sqrt{n})$ times smaller than $\min\{\omega,\tau\}$, which leads to big gains in iteration complexity. For a  detailed study of the properties of the SPCDM method applied to the AdaBoost problem we refer to a follow up work of Fercoq~\cite{Fer-ParallelAdaboost}.

\begin{figure}[!ht]
\centering
\includegraphics[width=23em]{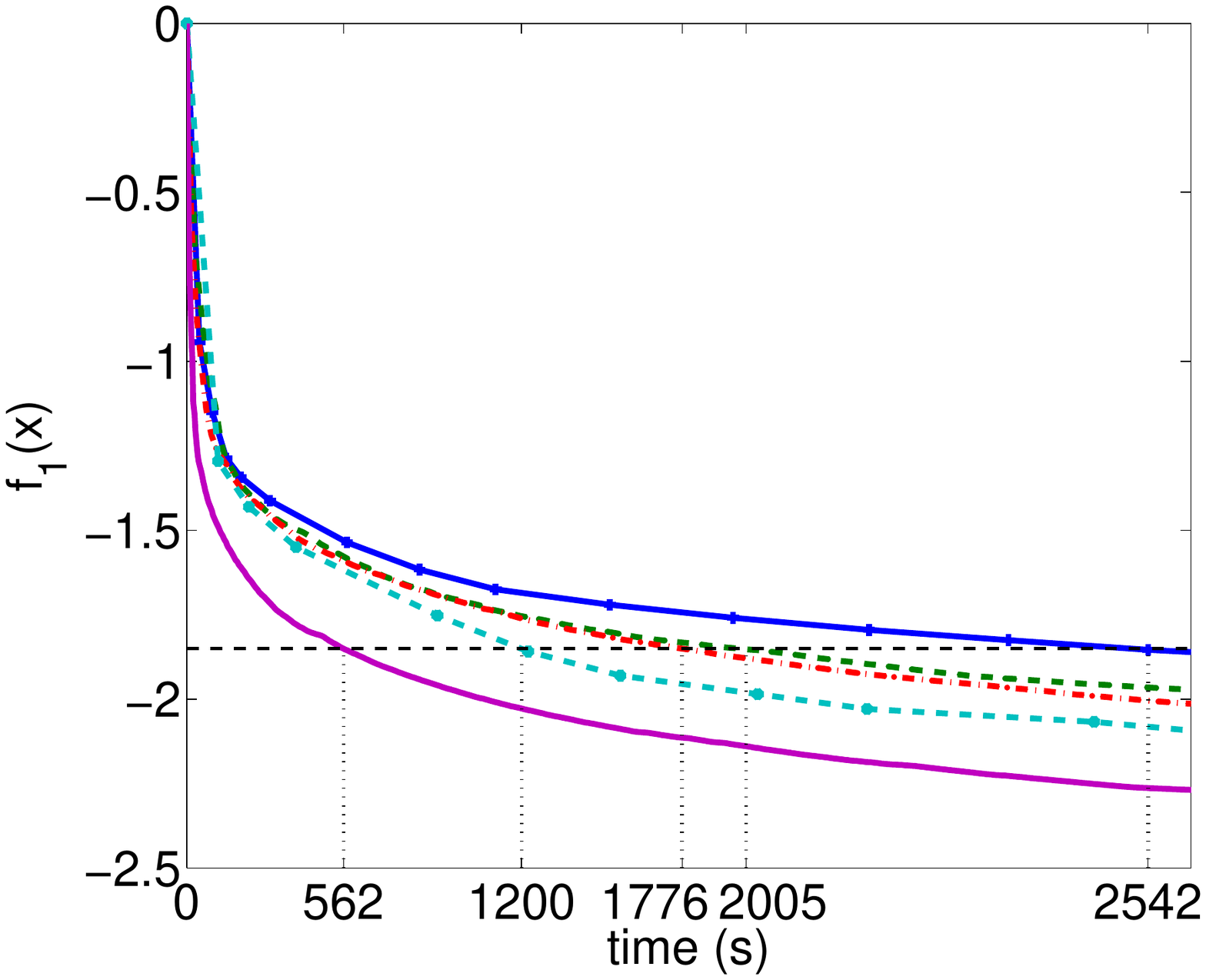}
\caption{Performance of the smoothed parallel coordinate descent method (SPCDM) with $\tau=1,2,4,8,16$ processors,
applied to the problem of minimizing the logarithm of the exponential loss \eqref{eq:log_exp_loss}, where $A \in \R^{m\times n}$ and $b\in \R^m$ are
given by the URL reputation dataset; $m=7,792,260$, $n=3,231,961$ and $\omega=414$. When $\tau=16$ processors were used, the method needed  562s to obtain a solution of a given accuracy (depicted by the horizontal line). When $\tau=8$ processors were used, the method needed 1200s, roughly double that time. Compared to a single processor, which needed 2542s, the setup with $\tau=16$ was nearly $5$ times faster. Hence, it is possible to observe nearly parallelization speedup, as our theory predicts.
Same colors were used as in Figure~\ref{fig:expe1norm}.}
\label{fig:expeadaboost}
\end{figure}

\textbf{Experiment.} In our last experiment we demonstrate how SPCDM (which can be viewed as a random parallel version of AdaBoost)  performs on the URL reputation dataset. Looking at Figure~\ref{fig:expeadaboost}, we see that parallelization leads to acceleration, and the time needed to decrease the loss to -1.85 is inversely proportional to the number of processors. Note that the additional effort done by increasing the number of processors from 4 to 8 is compensated by the increase of $\beta$ from $1.2$ to $2.0$ (this is the little step in the zoom of Figure~\ref{fig:compeso}). Even so, further acceleration takes place when one further increases  the number of processors.



{
\footnotesize
\bibliographystyle{plain}
\bibliography{spcdm}

\begin{thebibliography}{10}

\bibitem{PCDN_2013}
Yatao Bian, Xiong Li, and Yuncai Liu.
\newblock Parallel coordinate descent newton for large-scale l1-regularized
  minimization.
\newblock {\em arXiv1306:4080v1}, June 2013.

\bibitem{Bradley:PCD-paper}
Joseph~K. Bradley, Aapo Kyrola, Danny Bickson, and Carlos Guestrin.
\newblock Parallel coordinate descent for {L}1-regularized loss minimization.
\newblock In {\em 28th International Conference on Machine Learning}, 2011.

\bibitem{collins2002logistic}
Michael Collins, Robert~E. Shapire, and Yoram Singer.
\newblock Logistic regression, adaboost and bregman distances.
\newblock {\em Machine Learning}, 48(1-3):253--285, 2002.

\bibitem{LanBCD2013}
Cong~D. Dang and Lan Guanghui.
\newblock Stochastic block mirror descent methods for nonsmooth and stochastic
  optimization.
\newblock Technical report, Georgia Institute of Technology, September 2013.

\bibitem{Fer-ParallelAdaboost}
Olivier Fercoq.
\newblock Parallel coordinate descent for the {A}da{B}oost problem.
\newblock In {\em International Conference on Machine Learning and Applications
  - ICMLA'13}, 2013.

\bibitem{freund1995decision}
Yoav Freund and Robert~E. Shapire.
\newblock A decision-theoretic generalization of on-line learning and an
  application to boosting.
\newblock In {\em Computational Learning Theory}, pages 23--37. Springer, 1995.

\bibitem{guyon2004result}
Isabelle Guyon, Steve Gunn, Asa Ben-Hur, and Gideon Dror.
\newblock Result analysis of the {NIPS} 2003 feature selection challenge.
\newblock {\em Advances in Neural Information Processing Systems}, 17:545--552,
  2004.

\bibitem{Richtarik-GPower2010}
Michel Journ\'{e}e, Yurii Nesterov, Peter Richt\'{a}rik, and Rodolphe
  Sepulchre.
\newblock Generalized power method for sparse principal component analysis.
\newblock {\em Journal of Machine Learning Research}, 11:517--553, 2010.

\bibitem{Jaggi-ICML2013-block-Frank-Wolfe}
Simon Lacoste-Julien, Martin Jaggi, Mark Schmidt, and Patrick Pletcher.
\newblock Block-coordinate frank-wolfe optimization for structural svms.
\newblock In {\em 30th International Conference on Machine Learning}, 2013.

\bibitem{Leventhal:2008:RMLC}
Dennis Leventhal and Adrian~S. Lewis.
\newblock Randomized methods for linear constraints: Convergence rates and
  conditioning.
\newblock {\em Mathematics of Operations Research}, 35(3):641--654, 2010.

\bibitem{luxiao2013complexity}
Zhaosong Lu and Lin Xiao.
\newblock On the complexity analysis of randomized block-coordinate descent
  methods.
\newblock Technical report, Microsoft Research, 2013.

\bibitem{ma2009identifying}
Justin Ma, Lawrence~K. Saul, Stefan Savage, and Geoffrey~M. Voelker.
\newblock Identifying suspicious urls: an application of large-scale online
  learning.
\newblock In {\em Proceedings of the 26th Annual International Conference on
  Machine Learning}, pages 681--688. ACM, 2009.

\bibitem{BOOM2013}
Indraneel Mukherjee, Kevin Canini, Rafael Frongillo, and Yoram Singer.
\newblock Parallel boosting with momentum.
\newblock Technical report, Google Inc., 2013.

\bibitem{mukherjee2011rate}
Indraneel Mukherjee, Cynthia Rudin, and Robert~E. Shapire.
\newblock The rate of convergence of {A}da{B}oost.
\newblock {\em arXiv:1106.6024}, 2011.

\bibitem{Necoara:parallelCDM-MPC}
Ion Necoara and Dragos Clipici.
\newblock Efficient parallel coordinate descent algorithm for convex
  optimization problems with separable constraints: application to distributed
  mpc.
\newblock {\em Journal of Process Control}, 23:243--253, 2013.

\bibitem{Necoara:Coupled}
Ion Necoara, Yurii Nesterov, and Francois Glineur.
\newblock Efficiency of randomized coordinate descent methods on optimization
  problems with linearly coupled constraints.
\newblock Technical report, Politehnica University of Bucharest, 2012.

\bibitem{Necoara_composite-coupled}
Ion Necoara and Andrei Patrascu.
\newblock A random coordinate descent algorithm for optimization problems with
  composite objective function and linear coupled constraints.
\newblock Technical report, University Politehnica Bucharest, 2012.

\bibitem{Nesterov05:smooth}
Yurii Nesterov.
\newblock Smooth minimization of nonsmooth functions.
\newblock {\em Mathematical Programming}, 103:127--152, 2005.

\bibitem{Nesterov:2010RCDM}
Yurii Nesterov.
\newblock Efficiency of coordinate descent methods on huge-scale optimization
  problems.
\newblock {\em SIAM Journal on Optimization}, 22(2):341--362, 2012.

\bibitem{Nesterov-Subgrad-Huge}
Yurii Nesterov.
\newblock Subgradient methods for huge-scale optimization problems.
\newblock {\em CORE DISCUSSION PAPER 2012/2}, 2012.

\bibitem{Nesterov:2007composite}
Yurii Nesterov.
\newblock Gradient methods for minimizing composite function.
\newblock {\em Mathematical Programming}, 140(1):125--161, 2013.

\bibitem{palit2012scalable}
Indranil Palit and Chandan~K. Reddy.
\newblock Scalable and parallel boosting with {M}ap{R}educe.
\newblock {\em IEEE Transactions on Knowledge and Data Engineering},
  24(10):1904--1916, 2012.

\bibitem{RT:TTD2011}
Peter Richt\'{a}rik and Martin Tak\'{a}\v{c}.
\newblock Efficient serial and parallel coordinate descent methods for
  huge-scale truss topology design.
\newblock In {\em Operations Research Proceedings}, pages 27--32. Springer,
  2012.

\bibitem{RT:UCDC}
Peter Richt\'{a}rik and Martin Tak\'{a}\v{c}.
\newblock Iteration complexity of randomized block-coordinate descent methods
  for minimizing a composite function.
\newblock {\em Mathematical Programming}, 2012.

\bibitem{RT:SPARS11}
Peter Richt\'{a}rik and Martin Tak\'{a}\v{c}.
\newblock Efficiency of randomized coordinate descent methods on minimization
  problems with a composite objective function.
\newblock In {\em 4th Workshop on Signal Processing with Adaptive Sparse
  Structured Representations}, June 2011.

\bibitem{RT:PCDM}
Peter Richt\'{a}rik and Martin Tak\'{a}\v{c}.
\newblock Parallel coordinate descent methods for big data optimization
  problems.
\newblock {\em arXiv:1212.0873}, November 2012.

\bibitem{RTA:24am}
Peter Richt\'{a}rik, Martin Tak\'{a}\v{c}, and S.~Damla Ahipa\c{s}ao\u{g}lu.
\newblock Alternating maximization: unifying framework for 8 sparse {PCA}
  formulations and efficient parallel codes.
\newblock {\em arXiv:1212:4137}, December 2012.

\bibitem{Ruszczynski95}
Andrzej Ruszczy\'{n}ski.
\newblock On convergence of an augmented {L}agrangian decomposition method for
  sparse convex optimization.
\newblock {\em Mathematics of Operations Reseach}, 20(3):634--656, 1995.

\bibitem{Boosting-book}
Robert~E. Schapire and Yoav Freund.
\newblock {\em Boosting: Foundations and Algorithms}.
\newblock The MIT Press, 2012.

\bibitem{ShalevTewari09}
Shai Shalev-Shwartz and Ambuj Tewari.
\newblock Stochastic methods for $\ell_1$-regularized loss minimization.
\newblock {\em Journal of Machine Learning Research}, 12:1865--1892, 2011.

\bibitem{SSS2013-accelerated}
Shai Shalev-Shwartz and Tong Zhang.
\newblock Accelerated mini-batch stochastic dual coordinate ascent.
\newblock {\em arXiv:1305.2581v1}, May 2013.

\bibitem{Stoch-dual-Coord-Ascent}
Shai Shalev-Shwartz and Tong Zhang.
\newblock Stochastic dual coordinate ascent methods for regularized loss
  minimization.
\newblock {\em Journal of Machine Learning Research}, 14:567--599, 2013.

\bibitem{minibatch-ICML2013}
Martin Tak\'{a}\v{c}, Avleen Bijral, Peter Richt\'{a}rik, and Nathan Srebro.
\newblock Mini-batch primal and dual methods for {SVM}s.
\newblock In {\em 30th International Conference on Machine Learning}, 2013.

\bibitem{tao2012stochastic}
Qing Tao, Kang Kong, Dejun Chu, and Gaowei Wu.
\newblock Stochastic coordinate descent methods for regularized smooth and
  nonsmooth losses.
\newblock {\em Machine Learning and Knowledge Discovery in Databases}, pages
  537--552, 2012.

\bibitem{TRB2013:DQA}
Rachael Tappenden, Peter Richt\'{a}rik, and Burak B\"{u}ke.
\newblock Separable approximations and decomposition methods for the augmented
  {L}agrangian.
\newblock {\em arXiv:1308.6774}, August 2013.

\bibitem{TRG:InexactCDM}
Rachael Tappenden, Peter Richt\'{a}rik, and Jacek Gondzio.
\newblock Inexact coordinate descent: complexity and preconditioning.
\newblock {\em arXiv:1304.5530}, April 2013.

\bibitem{telgarsky2012primal}
Matus Telgarsky.
\newblock A primal-dual convergence analysis of boosting.
\newblock {\em The Journal of Machine Learning Research}, 13:561--606, 2012.

\end{thebibliography}

}

\end{document}